\documentclass[12pt]{article}
\pdfoutput=1
\usepackage{times}
\usepackage[margin=1in]{geometry}
\usepackage[onehalfspacing]{setspace}
\usepackage{xcolor}
\usepackage{listings}
\usepackage[T1,OT1]{fontenc}
\usepackage{textcomp}
\usepackage{tikz}
\usepackage{amsmath}
\usepackage{natbib}
\usepackage[hidelinks]{hyperref}
\usepackage[capitalize]{cleveref}
\usepackage{graphicx,amssymb}
\usepackage{amsthm}
\usepackage{float}
\usepackage{xfrac}
\usepackage{xspace}
\usepackage{enumerate,multicol,multirow}
\usepackage{cases}
\usepackage{caption,subcaption}

\usepackage{todonotes}

\DeclareMathOperator*{\argmax}{argmax}

\newtheorem{corollary}{Corollary}
\newtheorem{theorem}{Theorem}
\newtheorem{lemma}{Lemma}

\newtheorem{proposition}{Proposition}

\newtheorem{definition}{Definition}

\newtheorem{example}{Example}

\newcommand{\notshow}[1]{{}}
\newcommand{\AutoAdjust}[3]{{ \mathchoice{ \left #1 #2  \right #3}{#1 #2 #3}{#1 #2 #3}{#1 #2 #3} }}
\newcommand{\Xcomment}[1]{{}}

\newcommand{\InParentheses}[1]{\AutoAdjust{(}{#1}{)}}

\newcommand{\supp}{\mathbf{Supp}}

\newcommand{\given}{\,\vert\,}

\newcommand{\prob}[2][]{\text{\bf Pr}\ifthenelse{\not\equal{}{#1}}{_{#1}}{}\!\left[{\def\givenn{\middle|}#2}\right]}
\newcommand{\expect}[2][]{\text{\bf E}\ifthenelse{\not\equal{}{#1}}{_{#1}}{}\!\left[{#2}\right]}

\newcommand{\tparen}{\big}
\newcommand{\tprob}[2][]{\text{\bf Pr}\ifthenelse{\not\equal{}{#1}}{_{#1}}{}\tparen[{\def\given{\tparen|}#2}\tparen]}
\newcommand{\texpect}[2][]{\text{\bf E}\ifthenelse{\not\equal{}{#1}}{_{#1}}{}\tparen[{\def\given{\tparen|}#2}\tparen]}

\newcommand{\sprob}[2][]{\text{\bf Pr}\ifthenelse{\not\equal{}{#1}}{_{#1}}{}[#2]}
\newcommand{\sexpect}[2][]{\text{\bf E}\ifthenelse{\not\equal{}{#1}}{_{#1}}{}[#2]}

\newcommand{\lbr}[1]{\left\{#1\right\}}
\newcommand{\rbr}[1]{\left(#1\right)}

\newcommand{\indicate}[1]{\mathbf{1}\left[#1\right]}
\newcommand{\reals}{\mathbb{R}}

\newcommand{\val}{v}
\newcommand{\belief}{\nu}

\newcommand{\sigprob}{ r\InParentheses{s} }
\newcommand{\sigprobsingle}{ r }
\newcommand{\vals}{V}

\newcommand{\dist}{F}
\newcommand{\density}{f}
\newcommand{\util}{u}
\newcommand{\alloc}{x}
\newcommand{\pay}{p}
\newcommand{\experiment}{\sigma}
\newcommand{\signal}{s}
\newcommand{\signals}{S}

\newcommand{\mech}{M}

\newcommand{\rev}{{\rm Rev}}

\newcommand{\optrev}{{\rm OPT\text{-}Rev}}
\newcommand{\optwel}{{\rm OPT\text{-}Wel}}
\newcommand{\wel}{{\rm Wel}}

\setcitestyle{authoryear}

\providecommand{\keywords}[1]
{
  \small	
  \textbf{\textit{Keywords---}} #1
}

\providecommand{\JEL}[1]
{
  \small	
  \textbf{\textit{JEL---}} #1
}

\begin{document}

\begin{titlepage}

\title{Multi-Dimensional Screening with Endogenous Information Disclosure\thanks{The authors thanks S.\ Nageeb Ali, Dirk Bergemann, Yi-Chun Chen, Yannai Gonczarowski, Marina Halac, Eran Shmaya, Philipp Strack, Kai Hao Yang, Kun Zhang, and audiences at Harvard EconCS seminar, GAIMSS'25 Workshop and SAET Conference for helpful comments and suggestions. Yang Cai was supported by the NSF Awards CCF-1942583 (CAREER) and CCF-2342642. Yingkai Li thanks NUS Start-up Grant for financial support. Part of the work was done while Yingkai Li was a postdoc at Yale under the support of a Sloan Foundation Research Fellowship. Jinzhao Wu was supported by the NSF Awards CCF-1942583 (CAREER), CCF-2342642, and a CADMY Research Fellowship. The two-product optimality result and the improved three-fifths approximation for an arbitrary number of products were developed with assistance from ChatGPT Astra. All other research results were developed by the human authors.}}

\author{Yang Cai\thanks{Department of Computer Science, Yale University. Email: \url{yang.cai@yale.edu}} \\ 
\and Yingkai Li\thanks{Department of Economics, National University of Singapore. Email: \url{yk.li@nus.edu.sg}}
\and Jinzhao Wu\thanks{Department of Computer Science, Yale University. Email: \url{jinzhao.wu@yale.edu}} }

\date{}

\maketitle

\begin{abstract}
We study multi-product monopoly pricing where the seller jointly designs the selling mechanism and the information structure for a unit-demand buyer to learn his values. Unlike the case with exogenous information where no mechanism with bounded menu size can extract a non-negligible fraction of the optimal revenue \citep{hart2019selling}, we show that when the seller controls information, even uniform pricing guarantees at least half of the optimal revenue. Allowing two prices raises the guarantee to three-fifths of welfare, and hence of optimal revenue. We also show that the uniform-pricing guarantee degrades smoothly with the informativeness of the buyer's outside information. With two products, deterministic pricing attains the optimal revenue for every joint distribution. For any number of products, negatively affiliated or exchangeable values allow pricing to extract the full surplus. Our results highlight the power of information design in making pricing mechanisms approximately optimal in multi-dimensional settings.
\end{abstract}


\keywords{pricing, correlation, information design, approximation}

\JEL{D44,D82,D83}

\end{titlepage}

\section{Introduction}
\label{sec:intro}

When a single seller offers a portfolio of differentiated products, designing relative prices becomes the central lever for steering which alternatives attract demand and how revenue is earned across the firm's products. Relative to the single-product benchmark, pricing one product inevitably reallocates demand across the firm's own products, so internal substitution and cannibalization shape every pricing move. Classic theory developed the tools for pricing multiple products and structured menus \citep{armstrong_multiproduct_1996,rochet_ironing_1998,hart2019selling}. Related work also documents how relative prices and prominence govern take-up across available options \citep{edelman2007internet,athey2011position}. This paper starts with this canonical problem and studies how a platform's design of prices and information can be deployed jointly to manage substitution across its own products.

A platform typically does more than post prices: it controls what buyers learn and which options they consider through ranking and recommendations, the salience of attributes, and the coarseness of disclosure. These instruments are pervasive. Online travel intermediaries, such as Priceline and Hotwire, deploy opaque offers that disclose ratings, neighborhoods, or sample reviews while withholding identities until purchase; auto and housing markets rely on pre-purchase inspections or certifications that selectively reveal quality; app stores and streaming services shape attention through curation and previews well before any transaction. Such practices interact with pricing by influencing which products receive consideration and how demand is distributed across the firm's products. 

Motivated by this, recent work formulates a joint-design problem in which the seller commits ex ante to (i) a public menu and (ii) an information policy that governs what the buyer can learn about the buyer's valuations \citep[e.g.,][]{bergemann2022screening}. In this formulation, the buyer privately observes the realized signal, and prices and allocations do not directly condition on it (i.e., the menu is fixed ex ante; the buyer self-selects from it based on the signal).
Because information is chosen alongside the menu, the optimal mechanism is a finite menu aligned with the disclosure policy, with pooling used to manage rents even under rich heterogeneity.

We adopt the joint design perspective in a multi-product setting. The platform commits to a lottery-menu mechanism that offers randomized allocation across products and associated transfers, in parallel with an information policy that determines what each buyer privately learns about their own valuations before purchasing. We focus on environments where values for multiple products are realized jointly, allowing arbitrary correlation, and the buyer demands at most one product \citep[e.g.,][]{hart2019selling}.\footnote{We provide a brief discussion of additive valuations in \cref{subsec:additive}.} This contrasts with \citet{bergemann2022screening}, who studies a single-product environment with vertical quality differentiation \citep[c.f.,][]{mussa1978monopoly}. While we maintain their institutional constraint of posted menus without directly conditioning on realized private signals, the multi-product unit-demand setting introduces an additional economic force: horizontal information, alongside vertical information about overall appeal. Greater horizontal precision improves matching efficiency across substitutes, whereas coarser disclosure compresses information and limits rents; thus, the optimal policy balances these two considerations.

The optimal joint design problem in our paper is inherently multi-dimensional. Even when one instrument is held fixed, characterizing the other remains challenging \citep[e.g.,][]{kleiner2024extreme,daskalakis2017strong}. For instance, the optimal mechanism with an exogenous information structure is known only in special environments \citep[e.g.,][]{manelli2006bundling,thirumulanathan2019optimal,bikhchandani2022selling,haghpanah2021pure} or when robustness is an explicit concern for the monopolist \citep[e.g.,][]{carroll2017robustness}. 

The difficulty, however, is not confined to characterization. In the classic Bayesian framework, where the buyer's private information is exogenously given and values may be arbitrarily correlated, \citet{hart2019selling} show that multi-product screening is hard even for the purpose of approximation. Already with two products, the revenue gap between deterministic pricing and the optimal mechanism is unbounded; more strongly, no menu of bounded size---however the finitely many lottery options in it are chosen---guarantees any constant fraction of the optimal revenue. Approximate revenue maximization therefore requires both sophisticated randomization and unbounded menu complexity. This result is commonly read as delimiting the reach of simple mechanisms: positive results for simple mechanisms have largely been confined to independent or otherwise restricted value distributions.

The main contribution of our paper is to show that this worst-case failure is an artifact of the assumption that the buyer's information is exogenous. When the seller designs the information structure alongside the menu, the type space over which she screens becomes a design variable in its own right: instead of inheriting the buyer's multi-dimensional private information, the seller manufactures it. A disclosure policy that tells the buyer enough to identify a product worth buying, but little about how much he values it, collapses the screening problem into one that the simplest available mechanism can handle. Formally, we show that in the joint design problem, deterministic pricing mechanisms, when paired with an appropriately chosen information-disclosure policy, are competitive,\footnote{We say that a mechanism is competitive if it extracts at least an absolute constant fraction of the optimal revenue, independent of the number of products for sale or the underlying valuation distribution.} extracting at least three-fifths of the optimal revenue.

In particular, even uniform-pricing mechanisms, where all products are priced identically, extract at least half of the optimal welfare, and hence at least half of the optimal revenue (\cref{prop:uniform_pricing}).
This bound is tight because uniform pricing lacks the flexibility to tailor prices to heterogeneous priors across products. The key intuition is that uniform pricing helps counteract cannibalization stemming from unit-demand preferences, while a carefully chosen information structure allows the seller to capture most of the buyer's surplus with limited efficiency loss from not fully revealing product values. Our results indicate that, when the seller controls information structures, guaranteeing a large fraction of optimal revenue does not require complex lotteries or price discrimination across products.
Moreover, although our baseline model assumes that the buyer is initially uninformed, this assumption is not essential for the main economic force. In \cref{prop:partially_informed_buyer}, we show that the guarantee of uniform pricing degrades smoothly with the informativeness of the buyer's outside information. Thus, the relevant practical condition is not that the seller is the buyer's only source of information, but rather that the seller has a meaningful informational advantage through control over product disclosure.

Three features of the comparison with \citet{hart2019selling} are worth emphasizing. First, the improvement does not come from weakening the benchmark. Because the seller in our problem is always free to disclose values fully, the optimal revenue under joint design weakly exceeds the optimal revenue of the exogenous-information problem for the same distribution $\dist$. Second, the guarantee is uniform rather than distributional. The factor of two holds for every joint distribution of values and every number of products, whereas the exogenous-information gap is already unbounded with two products; the contrast is thus not between hard and easy value distributions, but between two informational environments over the same class of distributions. Third, informational complexity substitutes for menu complexity. The mechanism we construct posts a single price and uses a disclosure policy whose signals are simply sets of products guaranteed to contain the buyer's favorite. What an unbounded menu of lotteries was needed for under exogenous information is accomplished here by coarsening what the buyer learns.

The substitution between the two forms of complexity is starkest on the very instances that make the exogenous-information problem hard. In \cref{example:HartNisan} we revisit a simplified version of the construction used by \citet{hart2019selling}, in which the approximately optimal mechanism assigns each type a lottery codirectional with its value vector. The construction has bite because a fully informed buyer knows both which product he prefers and how intensely he prefers it, and the seller must screen the two jointly. A seller who controls disclosure can separate them: revealing only which product is best leaves the buyer uncertain about how much he values it, and a single posted price equal to the buyer's expected maximum value is then incentive compatible and extracts the \emph{full} surplus. On instances of the form used to show that no bounded menu guarantees a constant fraction of the optimal revenue, a posted price is thus exactly optimal. At the same time, we do not claim that endogenous information makes the joint design problem easy. \cref{exp:complex_info,exp:lottery_opt} show that optimal joint designs remain intricate: neither fully revealing which product is best nor revealing nothing is optimal in general, and lotteries can be strictly necessary for exact optimality once there are at least three products. What endogenous information removes is the failure of approximation, not the difficulty of characterization.

Allowing just two distinct prices raises the welfare guarantee from one-half to three-fifths (\cref{thm:approx_opt}). The seller pools some states to support sales at the higher price and uses the remaining states to support sales at the lower price. The choice of which products receive each price ensures that cheaper products do not undermine the high-price sales. Thus limited price differentiation improves the revenue guarantee while preserving a simple deterministic menu.

With two products, pricing is exactly revenue-optimal for every joint distribution (\cref{thm:two_item_optimal}). The seller can achieve this optimum with at most two signals, each recommending a purchase. Thus lotteries have no revenue advantage in the two-product joint design problem, even though they can be essential when information is fixed. Optimal revenue need not equal welfare: the tight welfare guarantee remains $\frac23$ (\cref{prop:two_item_approx}). With more products, lotteries can improve revenue, as \cref{exp:lottery_opt} shows, while negatively affiliated or exchangeable values allow pricing to extract the full surplus (\cref{thm:full_surplus_multi,thm:full_surplus_exchange}). 

Read alongside the platform applications described above, these results suggest a reinterpretation of the pessimism that correlated-value environments have traditionally invited. That pessimism is tied to a market in which buyers arrive already informed and the seller's only instrument is the menu. In markets where the seller or platform also curates what buyers learn---through ranking, opaque offers, previews, or certification---the prevalence of simple posted prices is not a puzzle to be explained away: coarse disclosure and simple pricing are complements, and together they are approximately optimal no matter how values are correlated. The economic content of our exercise is therefore not only that an additional instrument raises revenue, which is immediate, but that the instrument changes which selling mechanisms are worth using.

\subsection{Related Work}
\label{sub:literature}
Our paper relates to the literature on information design in auctions \citep[e.g.,][]{bergemann2007information,esHo2007optimal,li2017discriminatory,Bergemann2022Optimal,deb2021multi,schottmuller2023optimal,chen2023information,ravid2024monopoly,cai2024algorithmic,wei2024reverse,bergemann2022screening}. The most closely related papers are \citet{shi2024multi,bergemann2022screening,bergemann2025alignment}. \citet{bergemann2022screening} study a single-product monopoly environment with vertical quality differentiation, \`a la \citet{mussa1978monopoly}, and formulate a joint-design problem in which the seller chooses both a menu and an information policy. In contrast, we analyze multi-product monopoly pricing in the spirit of \citet{hart2019selling}, where horizontal differentiation and substitution across the seller's products are central to the joint design. 
\citet{shi2024multi} consider a related multi-product setting under independence across products; in our framework, independence is a special case of negative affiliation, for which the optimal design is tractable (see \cref{subsec:fullsurplus}). Beyond that special case, correlation is first-order in applications. For instance, in hotel pricing on online platforms, product values are naturally correlated through location, brand, and policies (e.g., smoking, cancellation). Finally, \citet{bergemann2025alignment} examine an oligopoly environment in which consumer surplus is aligned with total surplus, whereas our focus is on a monopoly platform and revenue.

The methodology we adopt in our paper is related to an extensive body of work in computer science that provides approximation guarantees for simple mechanisms in various settings \citep[e.g.,][]{hartline2009simple,chawla2010multi,cai2016duality,cai_simple_2017,babaioff2017menu,jin2019tight,babaioff2020simple,cai_simultaneous_2023,daskalakis_multi-item_2022,feng2023simple}.
This study of approximation guarantees for simple mechanisms has also recently gained traction in the economics literature \citep[e.g.,][]{bulow1989simple,hartline2015non,hart2017approximate,hart2019selling,mirrokni2020non,bergemann2022third,akbarpour2023algorithmic}. 
The main purpose of such approximation exercises is not so much to pinpoint the exact approximation factor, but rather to facilitate relative comparisons based on worst-case performance, thereby distilling the salient features of approximately optimal mechanisms across different settings.
For a more detailed discussion on interpreting approximation results in economic environments, see the surveys by \citet{hartline2012approximation} and \citet{roughgarden2019approximately}.

The most relevant approximation result in the literature is due to \citet{hart2019selling}, who show that in the absence of information design, even selling just two correlated products to a buyer can lead to an unbounded revenue gap between the optimal mechanism and deterministic pricing mechanisms. In contrast, we provide the first constant-factor approximation guarantee for deterministic pricing mechanisms in general correlated value environments by considering settings in which the seller can design the information structures. For two products, the guarantee becomes exact: deterministic pricing attains the joint optimum without any restriction on correlation.
In concurrent work, \citet{wang2026replacing} also establishes the revenue optimality of item pricing for two products with finitely supported values. Our two-product result allows arbitrary integrable priors, and our approximation guarantees apply to any number of products.

The conceptual idea that endogenous information can make simple mechanisms competitive has also been observed in other auction environments, where endogeneity arises because buyers acquire information optimally. For example, with optimal buyer learning, \citet{deb2021multi} show that bundling is optimal for selling multiple products with additive valuations under the exchangeable prior assumption, while \citet{li2022selling} shows that selling full information using posted pricing yields a 2-approximation to the optimum when the buyer can endogenously acquire additional information from other sources. In contrast, our model is substantially different: the seller controls the information structure, and the simple mechanisms we consider are pricing mechanisms.

\section{Model}
\label{sec:model}

We consider the problem of selling $m$ heterogeneous products to a buyer. 
The buyer has a value $\val = (\val_1,\dots,\val_m)$ drawn from a potentially correlated distribution $\dist$, 
where $\val_i\in \reals_{+}$ is the buyer's value for product $i$ for any $i\in[m]$.
We also refer to $\val$ as the buyer's type.
Let $\vals\subseteq \reals^m_{+}$ be the support of distribution~$\dist$. 
To simplify the exposition, we assume that $\dist$ is a discrete distribution with finite support.\footnote{All results in our paper extend to continuous distributions as well. }
Let $\density$ be the probability mass function. 
For any product $i$, let $\dist_i$ be the marginal distribution over values for product $i$ and let $\vals_i$ be the support of distribution~$\dist_i$. 
We assume that $\dist$ is common knowledge. 

\paragraph{Information Structures}
In contrast to classic monopoly pricing settings where the buyer is privately informed about his values for the products, 
we focus on a baseline model where the buyer is initially uninformed and relies on the information structure designed by the seller to learn his values.\footnote{We relax this assumption in \cref{subsec:partially_informed_buyer}, where the buyer may also receive an exogenous private signal from an outside source.} 
Specifically, the seller can commit to an information structure $(\signals,\experiment)$ 
where~$\signals$ is a measurable signal space
and $\experiment: \vals \to \Delta(\signals)$ is a mapping from values to signals. 
The buyer privately observes the realized signal $\signal\in\signals$
and updates his belief according to Bayes' rule. 

For any signal $\signal\in\signals$, let \( \sigprob \) denote the ex-ante probability of receiving a specific signal~$\signal$, and \( \belief(\signal) \) represent the buyer's posterior mean values upon receiving signal~$\signal$. 
Furthermore, \( \belief_i(\signal) \) specifies the \(i\)-th coordinate in \( \belief(\signal) \).
That is, 
\begin{align*}
   \sigprob = \expect[\val \sim \dist, \signal' \sim \experiment\InParentheses{\val}]{\indicate{\signal' = \signal}} \text{ and }  
   \belief_i\InParentheses{\signal} = \expect[\val \sim \dist, \signal' \sim \experiment\InParentheses{\val}]{\val_i \given \signal' = \signal} \text{ for all } i\in [m]. 
\end{align*}

\paragraph{Buyer's Utilities}
In this paper, we assume that the buyer has unit demand for the products. 
Moreover, after receiving the signal from the mechanism, the buyer must make a consumption choice before the realization of the values to enjoy the utility.\footnote{For instance, when the buyer books hotel rooms for his travel destination, the values are realized only after staying in a single hotel room for the night. This way of modeling unit-demand utilities is consistent with \citet{hart2019selling} if we interpret the private types in their model as the posterior mean values of the buyer. We also provide additional discussions for an alternative unit-demand utility model in \cref{subapx:unit_demand}.} 
That is, for any realized allocation $\alloc\in\{0,1\}^m$, 
the utility of the buyer with a posterior mean value $\belief$ for receiving allocation $\alloc$ while paying a price $\pay$ is 
\begin{align*}
\util(\alloc,\pay;\belief) = \max_{i\in[m]} \belief_i\alloc_i - \pay.
\end{align*}
Given such unit-demand utility of the buyer, it is without loss of generality to focus on mechanisms with (random) allocation $x\in [0,1]^m$ such that $\sum_i x_i\leq 1$.
Moreover, for any such (random) allocation $x\in [0,1]^m$, the expected utility of the buyer is
\begin{align*}
\util(\alloc,\pay;\belief) = \sum_{i\in[m]} \belief_i\alloc_i - \pay
= \expect[v\sim\belief]{\sum_{i\in[m]} v_i\alloc_i} - \pay.
\end{align*}
Finally, we omit $\belief$ from the notation when the buyer's value is clear from the context. 

\paragraph{Mechanisms}
{Under an information structure \(\rbr{\signals, \experiment}\)}, a mechanism $\mech=(\alloc,\pay)$ is a pair containing 
an allocation rule $\alloc : \signals \to [0,1]^m$,
and a payment rule $\pay : \signals \to \reals$. 
Mechanism~$\mech$ is 
\emph{incentive compatible} (IC) if 
\[\expect[v\sim F,s\sim \sigma(v)]{\util(\alloc(\signal), \pay(\signal); \belief(\signal)) \given \signal}\ge \expect[v\sim F,s\sim \sigma(v)]{\util(\alloc(\signal'), \pay(\signal'); \belief(\signal)) \given \signal}\]
for all $\signal,\signal'\in\signals$.
Moreover, this mechanism is \emph{individually rational} (IR) if
\begin{align*}
\expect[v\sim F,s\sim \sigma(v)]{\util(\alloc(\signal), \pay(\signal); \belief(\signal)) \given \signal} \geq 0
\end{align*}
for all $\signal\in\signals$. 
To simplify the notation, we use $\expect[\dist,\experiment]{\cdot}$ to denote $\expect[v\sim F,s\sim \sigma(v)]{\cdot}$, and we also omit $\dist$ in $\expect[\dist,\experiment]{\cdot}$
when it is clear from the context. 

For any IC-IR mechanism $\mech$, we denote its expected revenue as 
\begin{align*}
\rev(\mech) 
= \expect[\dist,\experiment]{\pay(\signal)}
\end{align*}
and its expected welfare as 
\begin{align*}
\wel(\mech) = \expect[\dist,\experiment]{\sum_{i\in[m]}\val_i\cdot\alloc_i(\signal)}. 
\end{align*}
The optimal revenue and optimal welfare are denoted as 
\begin{align*}
\optrev = \max_{\rbr{\signals,\experiment}\atop \mech \text{ is IC-IR}} \rev(\mech) 
\quad\text{and}\quad
\optwel = \max_{\rbr{\signals,\experiment}\atop \mech \text{ is IC-IR}} \wel(\mech).
\end{align*}

By the revelation principle, it is without loss of generality to focus on IC-IR mechanisms. 
The objective of the seller is to maximize the expected revenue. 

\paragraph{Deterministic Pricing}
Among all possible incentive-compatible and individually rational mechanisms, a particularly interesting and practically implementable class is pricing mechanisms. Intuitively, after committing to an information structure, pricing mechanisms assign deterministic prices to each product and allow the buyer to select his most preferred product based on the private signal received. Furthermore, a pricing mechanism is referred to as a uniform pricing mechanism if the prices for all products are identical. Formally, for any $i \in [m]$, let $\delta_i$ denote the point mass distribution corresponding to the sale of product $i$.
\begin{definition}[Deterministic Pricing]
{Given an information structure \(\rbr{\signals, \experiment}\)}, a mechanism $\mech=(\alloc,\pay)$ is implemented by a \emph{pricing mechanism} if there exists a vector of deterministic prices $(\hat{p}_1,\dots,\hat{p}_m)\in\reals^m_+$ such that for any $\signal \in \signals$, 
$\alloc(\signal) \in \argmax_{x\in\{0,1\}^m} \util(\alloc,\hat{p}\cdot\alloc;\,\belief(\signal))$
and $\pay(\signal) = \hat{p}\cdot\alloc(\signal)$.
Moreover, mechanism $\mech$ is a \emph{uniform pricing} mechanism if, in addition, there exists $\pay^*$ such that $\hat{p}_j=\pay^*$ for all $j\in[m]$.
\end{definition}
\noindent {It is clear that any mechanism implemented by a pricing mechanism satisfies IC and IR constraints.} To avoid ambiguity, we refer to the general randomized mechanism as the \emph{lottery mechanism}. 

\paragraph{Timeline}
We outline the timeline of our model:

\begin{enumerate}
    \item The seller selects an information structure \((\signals, \experiment)\) based on distribution $\dist$, where \(\experiment:\vals \rightarrow \Delta(\signals)\) maps values to signals. The seller then designs a mechanism \(\mech = (\alloc, \pay)\) on top of the information structure, where \(\alloc: \signals \rightarrow [0,1]^m\) is the allocation rule and \(\pay: \signals \rightarrow \mathbb{R}\) is the payment rule.
    \item The value profile \(\val\) is drawn from \(\dist\) and is unknown to both parties. A signal \(\signal \sim \experiment(\val)\) is \emph{privately} revealed to the buyer.
    \item The buyer reports a signal \(\signal'\). The final allocation is determined by \(\alloc(\signal')\), and the buyer must make the corresponding payment \(\pay(\signal')\).
\end{enumerate}

{In our model, as in \citet{bergemann2022screening}, the key institutional constraint is that the seller commits ex ante to a menu and an information policy, but the menu itself cannot directly condition on the realized signal.
This constraint is natural either because (i) the seller commits to the disclosure policy but does not observe the realized message at the individual level (e.g., buyer-only inspection/certification reports), or because (ii) the seller could observe it but cannot condition prices on it without running afoul of institutional limits on personalized pricing (e.g., legal, reputation or contractual discipline).
}

\section{Illustrative Examples}
\label{sec:example}
In this section, we first provide two examples to illustrate the complications involved in solving the optimal joint design problem. These examples further motivate our focus on the approximate optimality of pricing mechanisms. Then, we use another example from \citet{hart2019selling} to show that the approximation guarantee of pricing mechanisms can be significantly improved under endogenous information structures. 

\paragraph{Complexity of Optimal Mechanisms}
In the joint design problem, the revenue of the seller is the difference between the allocation efficiency and the information rent enjoyed by the buyer. Thus, there are two intuitive ways to extract high revenue from the buyer: (1) reveal horizontal information to the buyer to sustain an efficient allocation; or (2) reveal coarse information to the buyer such that he has no information rent. 
\cref{exp:complex_info} shows that it is possible for neither option to be optimal for the seller. 

\begin{example}\label{exp:complex_info}
There are two products for sale, and the buyer has three possible types drawn from the following distribution:
\begin{equation*}
(v_1, v_2) =
\begin{cases}
  (10,5) & \text{with probability } 0.2, \\
  (6,5) & \text{with probability } 0.4,\\
  (0,3) & \text{with probability } 0.4.
\end{cases}
\end{equation*}
\end{example}

Specifically, if the seller is limited to mechanisms that induce efficient allocation, the seller must separate type $(0,3)$ from types $(10,5)$ and $(6,5)$. 
In this mechanism, the optimal information structure is to pool types $(10,5)$ and $(6,5)$ together. The expected revenue in this case is $4.4$ by setting prices $\frac{16}{3}$ and $3$ on products $1$ and $2$ respectively.
On the other hand, if the seller chooses a mechanism that reveals coarse information to leave no information rent for the buyer, this mechanism would pool all types together. The posterior mean values are $4.4$ and $4.2$ for products $1$ and $2$, respectively, which leads to an expected revenue of $4.4$.

In contrast, the seller can improve her expected revenue by pooling types $(6,5)$ and $(0,3)$ and posting prices of $9$ and $4$ for products $1$ and $2$, respectively. 
This distorts the efficient allocation as the mechanism sells product $2$ to type $(6,5)$, and it leaves the buyer with an interim utility of $1$ for type $(10,5)$.
However, the expected revenue in this case is $5$, which is strictly higher than $4.4$.

The main intuition in this example is that, although pooling types $(6,5)$ and $(0,3)$ leads to inefficient allocations, the efficiency loss is mild since type $(6,5)$ holds a similar value for both products. However, such pooling increases the buyer's posterior mean value of product~$2$. As a result, the seller can post a higher price on product $2$ for selling to the pooled types, which reduces the information rent of type $(10,5)$. The latter benefit outweighs the loss of efficiency in this example. 
Finally, further pooling of $(10,5)$ and $(0,3)$ is not desirable, since their values for the two products are very different. The reduction in information rent is not sufficient to cover the efficiency loss in this latter case. 

In the next example, we show that the revenue optimal mechanisms may also involve complex lotteries, even under endogenous information structures. 

\begin{example}\label{exp:lottery_opt}
There are three products for sale, and the buyer has two possible types drawn from the following distribution:
\begin{equation}
    (v_1, v_2, v_3) =
    \begin{cases}
      (0,20,9) & \text{with probability } \frac{1}{2}, \\
      (4,0,5) & \text{with probability } \frac{1}{2}.
    \end{cases}
  \end{equation}
\end{example}
We show that lotteries are very effective in this example by fully revealing the values to the buyer. 
Specifically, by offering product $2$ at a price of $20$ and a lottery that sells products $1$ and $3$ with equal probabilities at a price of $4.5$, the seller can extract revenue of $12.25$. 
However, if the seller is limited to pricing mechanisms, regardless of the information structure, either (1) the price of product 3 is low, and hence type $(0,20,9)$ enjoys a large information rent; or (2) the allocation for type $(4,0,5)$ is sufficiently inefficient. In either case, the expected revenue of the pricing mechanism is less than 12, which is suboptimal compared to the lottery mechanism. The details of the calculations for \cref{exp:lottery_opt} are provided in \cref{sub:Suboptimality_of_Item_Pricing}.

\paragraph{Better Approximation for Pricing Mechanisms}
The previous examples illustrate the complications involved in characterizing the optimal mechanisms. 
The main focus of our paper is to show that pricing mechanisms are approximately optimal in this joint design problem without knowing the structure of the optimal mechanisms in general. 
To illustrate the idea that endogenous information can make pricing mechanisms competitive, we present an example using the ideas from \citet{hart2019selling}.
Constructions of the form below are the canonical hard case for multi-product screening with correlated values. With an exogenously informed buyer, \citet{hart2019selling} use such a family to show that the optimal revenue cannot be approximated to within any constant factor by a menu of bounded size, so that deterministic pricing---a menu with one entry per product---is arbitrarily far from optimal. We show that, once the seller controls disclosure, the same geometry makes a single posted price extract the full surplus.


\begin{example}
\label{example:HartNisan}
There are two products for sale. Let \(\varepsilon\in(0,1)\) be a sufficiently small parameter. 
Let \(\{(x_i,y_i)\}_{i=0}^n \) be a sequence of points sit on rays with evenly spaced angles and geometrically growing radii, i.e., $(x_i,y_i) = \rbr{\rbr{\frac{1}{\varepsilon}}^i \cos\rbr{\frac{i}{2n}\pi}, \rbr{\frac{1}{\varepsilon}}^i\sin\rbr{\frac{i}{2n}\pi}}$. The type distribution $\dist$ satisfies
\[
v = (x_i,y_i) \ \text{ and } \ v = (y_i,x_i)
\quad\text{each with probability}\quad
\frac{\varepsilon^i - \varepsilon^{i+1}}{2\rbr{1-\varepsilon^{n+1}}},
\]
so that $\dist$ is symmetric across the two products.

\begin{figure}[t]
\centering
\begin{tikzpicture}[scale=3,>=latex]

\begin{scope}[xshift=0cm]
  \draw[->] (0,0) -- (1.2,0) node[below] {$x$};
  \draw[->] (0,0) -- (0,1.2) node[left] {$y$};

  \draw[very thick] (1,0) arc[start angle=0, end angle=90, radius=1];

  \draw[dashed] (0,0) -- (1.1,1.1) node[above right] {$x=y$};

  \foreach \a in {0,10,20,30,40}{
    \fill[blue] ({cos(\a)},{sin(\a)}) circle[radius=0.015];
  }
  \fill[black] ({cos(45)},{sin(45)}) circle[radius=0.018];
  \foreach \a in {50,60,70,80,90}{
    \fill[red] ({cos(\a)},{sin(\a)}) circle[radius=0.015];
  }

  \node[blue] at (0.83,0.23) {\scriptsize $x>y$};
  \node[red]  at (0.23,0.87) {\scriptsize $x<y$};
\end{scope}

\draw[->,very thick] (1.45,0.6) -- node[above]{\scriptsize pooling} (2.25,0.6);

\begin{scope}[xshift=2.5cm]
  \draw[->] (0,0) -- (1.2,0) node[below] {$x$};
  \draw[->] (0,0) -- (0,1.2) node[left] {$y$};

  \draw[dashed] (0,0) -- (1.1,1.1) node[above right] {$x=y$};

  \fill[blue] (0.85,0.35) circle[radius=0.03]
    node[below right=2pt] {\scriptsize aggregate of $x>y$};
  \fill[red]  (0.35,0.85) circle[radius=0.03]
    node[above =2pt]  {\scriptsize \hspace{10pt} aggregate of $x<y$};
\end{scope}
\end{tikzpicture}
\caption{Weighted value profile before and after pooling. }
\label{fig:weight_profile}
\end{figure}

\end{example}

The exact construction in \citet{hart2019selling} is considerably more involved. In our simplification, we specify a distribution via a single sequence of point masses such that their probability-weighted values lie on an arc, whereas \citet{hart2019selling} employ arcs of varying radii with a more delicate parameterization. Nevertheless, this simplified construction suffices to illustrate the main insight.

In particular, \citet{hart2019selling} show that the (approximately) optimal mechanisms for instances of this form require complex second-degree price discrimination, where the seller assigns a corresponding lottery allocation that is codirectional with each value profile. 
However, observe that the weighted value profiles are approximately evenly spaced on an arc of a unit circle: 
\[
\begin{aligned}
\prob[v\sim \dist]{v = (x_i,y_i)} \cdot (x_i,y_i) &\approx \tfrac{1}{2}\rbr{\cos\rbr{\tfrac{i}{2n}\pi},\ \sin\rbr{\tfrac{i}{2n}\pi}},\\
\prob[v\sim \dist]{v = (y_i,x_i)} \cdot (y_i,x_i) &\approx \tfrac{1}{2}\rbr{\sin\rbr{\tfrac{i}{2n}\pi},\ \cos\rbr{\tfrac{i}{2n}\pi}},
\end{aligned}
\]
so the weighted profiles fill the arc symmetrically about the ray $x=y$.
This ensures that by only disclosing horizontal information to the buyer, i.e., informing the buyer which product he prefers, 
the posterior mean value exhibits a negative correlation, as we illustrate in \Cref{fig:weight_profile}.
Concretely, the two signals are equally likely, and by symmetry the posterior mean of a product conditional on it being the buyer's favorite equals $\expect[\dist]{\max_i \val_i}$ for both products, while the posterior mean of the other product is strictly lower. Posting the \emph{single} uniform price $\pay = \expect[\dist]{\max_i\val_i}$ therefore leaves the buyer no rent and sells with probability one, so the seller extracts the full surplus.
In \cref{lem:full_surplus}, we also provide a generalization of this observation to establish a sufficient condition for full surplus extraction using pricing mechanisms; the symmetry imposed above is exactly what makes $\dist$ exchangeable, so this instance is a special case of \cref{thm:full_surplus_exchange}. 
Essentially, the high level idea here is that the optimally designed information structure eliminates the need to use lotteries to achieve a more efficient allocation while maintaining a low information rent for the buyer. 

\section{Approximate Optimality of Pricing Mechanisms}
\label{sec:approximation}
Multidimensional screening is notoriously challenging. In environments with exogenous information structures, \citet{hart2019selling} show that when product values are correlated, even when there are only two products, approximating the optimal revenue remains hard.
Specifically, any mechanism that guarantees an $\epsilon$ fraction of the optimal revenue for any arbitrarily small constant $\epsilon > 0$ requires an unbounded number of menu options, each potentially involving complex lottery allocations.
This implies that simple deterministic pricing mechanisms can perform arbitrarily poorly relative to the optimal in correlated valuation environments with exogenous information (see \cref{prop:without_info} in \cref{subapx:without_info}).
 
In this section, we show that pricing mechanisms are approximately optimal when the seller can design the information structure through which the buyer learns his values, for any number of products and under arbitrary correlation. This stands in sharp contrast to the negative result of \citep{hart2019selling}.
This comparison illustrates a central insight of our paper: the ability to design information structures makes pricing mechanisms competitive in multi-product settings with correlated values.

\subsection{Approximately Optimal Information Structures}
\label{subsec:infor_structure}
We first introduce a class of information structures useful for our approximation results, which we term (coarse) horizontal disclosure. Although this class may not be optimal for the joint design problem, it is intuitive, and we show that the monopoly seller can restrict attention to it while still securing a good revenue guarantee. 

Specifically, let $i^*(\val)$ denote the product with the highest value under the valuation profile~$\val$, with ties broken arbitrarily.

\begin{definition}[(Coarse) Horizontal Disclosure]
\label{def:coarse_horizontal}
An information structure $(\signals,\experiment)$ is a \emph{horizontal disclosure} if $\signals = [m]$
and $\experiment(\val) = \delta_{i^*(\val)}$ for any $\val$, 
where $\delta_{i^*(\val)}$ is a deterministic distribution with realization $i^*(\val)$. 
An information structure $(\signals,\experiment)$ is a \emph{coarse horizontal disclosure} if $\signals = 2^{[m]}\setminus\{\emptyset\}$, and there exists a mapping $\tau: [m]\to \Delta(\signals)$ such that $\experiment(\val) = \tau(i^*(\val))$ for all~$\val$ and 
\[
\supp\big(\tau(i)\big)\subseteq \{\,s\in\signals:\; i\in s\,\}\;\;\text{for every }i\in[m].
\]
\end{definition}
In horizontal disclosure, the buyer is informed about the product that holds the highest value for them, with ties broken arbitrarily. Intuitively, under horizontal disclosure, the seller informs the buyer which product he prefers the most while keeping the buyer ignorant of the magnitude of his preference. 
It allows the buyer to discern the relative qualities among different products to improve allocation efficiency while keeping vertical information to a minimum in order to reduce information rent. 

Nonetheless, horizontal disclosure may still provide excessive information to the buyer due to correlations and, hence, leave a large information rent and low revenue.\footnote{Essentially, when there are correlations, revealing horizontal information may inevitably disclose undesirable vertical information as well, as we illustrate in the example in \cref{subapx:horizontal}.} In \cref{subapx:horizontal}, we provide a concrete example in correlated environments to show that the revenue under horizontal disclosure is almost negligible compared to the optimal. Thus, to achieve a good approximation guarantee, the information provided to the buyer sometimes needs to be coarser. 

Our definition of coarse horizontal disclosure addresses this issue by allowing the seller to disclose a set of products rather than just a single product. In this setting, the disclosed set is guaranteed to contain the buyer's most valuable product; however, it may also include others. Thus, the signal space consists of subsets of products, and the buyer is assured that his top-valued product lies within the disclosed set.
This provides less information to the buyer in the sense that any coarse horizontal disclosure is less informative than horizontal disclosure according to the Blackwell order. 
As an illustration, both horizontal disclosure (where each disclosed set is a singleton set) and revealing no information (where the disclosed set is the full set) are special cases of coarse horizontal disclosure, in which horizontal disclosure is the most informative, and revealing no information is the least informative.
We will show that coarse horizontal disclosure is sufficient to guarantee a good approximation to the optimal under any correlation structure. 

\subsection{Tight Approximation of Uniform Pricing}
We first establish an approximation guarantee for uniform pricing. This special case of pricing mechanisms has an appealing practical feature: the seller posts a single price; if the buyer accepts, they may purchase any product they prefer based on the disclosed information. We show that uniform pricing attains at least one-half of the optimal revenue in the worst case, and that this bound is tight. Thus, securing a large fraction of optimal revenue does not require complex lotteries. Indeed, a single uniform price suffices despite ex-ante heterogeneity across products.

\begin{theorem}[Approximate Optimality of Uniform Pricing]
\label{prop:uniform_pricing}
For any number of products and any (potentially correlated) distribution $\dist$ over values, there exists a uniform pricing mechanism $\mech$ with a coarse horizontal disclosure such that the expected revenue is at least half of the optimal revenue, i.e., $\rev(\mech) \geq \frac{1}{2} \cdot \optrev$. Moreover, there are instances where a revenue loss of half is necessary for uniform pricing mechanisms. 
\end{theorem}

To prove this result, instead of directly comparing it to the optimal (and hard-to-characterize) revenue, we analyze a more tractable upper bound: the optimal welfare. That is, we will show that $\rev(\mech) \geq \frac{1}{2} \cdot \optwel$. The welfare benchmark is easier to compute and allows for a clean approximation argument. Later, we will show that the approximation guarantee is tight, even against the revenue benchmark. 

As a warm up, consider first a simpler setting with a single product. Without control over information, the seller may be forced to offer the buyer significant information rents to satisfy incentive constraints. For instance, suppose the buyer's value is drawn from the distribution with CDF $F(v) = 1 - \frac{1}{v}$ on $[1, H)$, and $F(v) = 1$ for $v \geq H$. If the value is fully revealed to the buyer, the optimal price is $1$, yielding revenue $1$ and leaving the buyer utility $\ln H$. But if the seller reveals no information and charges the buyer the expected value (i.e., $1 + \ln H$), the seller captures all the surplus. The ability to design information structures, in this case, dramatically improves revenue, especially when $H$ is large.

With multiple products, the logic described above does not directly apply. If the seller discloses no information, the buyer does not know which product to choose, leading to potential misallocation and significant loss of efficiency, especially when the number of products is large. To achieve a good approximation of welfare, the seller must reveal certain information to guide the buyer toward high-value matches. Our results will show that there is a way to partially reveal information such that the allocation efficiency remains close to optimal welfare while almost all the rent can be extracted by the seller. 

Here's the key insight: by restricting attention to uniform pricing, there is no product-level price discrimination; thus, the seller is indifferent about which product is sold. This removes cannibalization concerns and simplifies the problem: the seller just needs to maximize the probability of sale at a fixed price.

Define $p^*$ as the highest uniform price for which there exists an information structure that guarantees the probability of selling one of the products to the buyer is equal to one. Intuitively, the information structure pools types whose highest value exceeds price~$p^*$ with those whose highest value falls below~$p^*$. In fact, it is sufficient to consider coarse horizontal disclosure for this purpose. 

Note that the expected revenue, given the price $p^*$, is also $p^*$. Two critical facts underpin the analysis:
\begin{enumerate}
    \item Types with a maximum value strictly above $p^*$: The mechanism extracts nearly the full surplus from these types, as any positive utility of the buyer would imply an opportunity to raise the price while maintaining a sale probability of 1 by further pooling in lower-value types, thereby increasing revenue. This implies that their welfare contribution is at most the revenue of $p^*$.
    \item Types with a maximum value weakly below $p^*$: Their welfare contribution is upper bounded by $p^*$, since no product exceeds this price.
\end{enumerate}

Putting these together: the revenue at price $p^*$ is at least half of the total welfare. This guarantees a 2-approximation to the welfare---and hence to the optimal revenue---using a simple uniform price mechanism. In essence, information design allows uniform pricing mechanisms to extract more rent from high-value types, but this comes at the expense of allocation efficiency by not selling to types whose maximum value is below the uniform price. Fortunately, the loss in the latter is proven to be not too large. 
The details of the proof are provided in \cref{apx:proofs}.





\paragraph{Comparison with Exogenous Information.}
It is useful to contrast \cref{prop:uniform_pricing} with the exogenous-information benchmark in \citet{hart2019selling}. 
In their setting, the buyer's values are already privately known, and the seller must screen over the resulting multidimensional type space. 
With arbitrary correlation across products, this type space can make deterministic pricing arbitrarily worse than the optimal mechanism, and approximate revenue maximization may require complex lottery menus. 
The conclusion of \cref{prop:uniform_pricing} is different not because the underlying correlation problem disappears, but because the seller can shape the buyer's posterior type before the buyer chooses from the menu. 
The coarse horizontal disclosure used in the proof pools value profiles so that the buyer receives enough information to identify a suitable product, but not so much vertical information that large information rents must be left to the buyer. 
Thus, information design changes the effective screening problem faced by the pricing mechanism.

Two further aspects of the comparison deserve emphasis, since they determine how much the contrast should be read into.
First, the benchmark is not weakened. The seller in our problem may always choose the fully revealing information structure, so $\optrev$ weakly exceeds the optimal revenue of the exogenous-information problem for the same distribution $\dist$; and the proof above bounds the revenue of uniform pricing against $\optwel$, which weakly exceeds $\optrev$ in turn. The improvement is thus not the mechanical consequence of comparing a more powerful simple mechanism against an unchanged optimum: the additional instrument is available to the optimal mechanism as well, and it raises the target at the same time as it raises the guarantee.
Second, the factor of two that remains is not a residue of correlation. In the tight example below, the buyer's value for the second product is deterministic, so the two products are uncorrelated. What defeats a single price there is that no pooling can raise the posterior mean of the deterministic product above its constant value, while the other product is worthless to the frequent type; the seller must therefore choose between serving a rare high-value type and a frequent low-value type. Allowing a second price already improves the worst-case guarantee (\cref{thm:approx_opt}), and in that same example pricing mechanisms extract the full surplus. What limits uniform pricing once information is designed is thus the inability to differentiate prices across ex-ante heterogeneous products, not the correlation among them. Note that with a single product the corresponding loss vanishes entirely: the seller discloses nothing and charges the prior mean, extracting the full surplus.

\paragraph{Practical Implications.}
At the same time, the baseline model should not be interpreted as requiring the buyer to have no information in every practical application. 
The fully uninformed buyer assumption is a useful benchmark because it isolates the force of seller-designed information, but many markets are better described by an intermediate case: the buyer may have some outside information, while the seller or platform remains better informed about product attributes, match quality, or the mapping from products to buyer values. 
This informational advantage is common in applications such as online travel platforms, housing and auto markets, app stores, and recommendation platforms, where the seller can control the granularity, timing, and salience of information even though the buyer may also consult outside sources.

The partially informed buyer result in \cref{prop:partially_informed_buyer} formalizes this interpretation. 
It shows that the guarantee in \cref{prop:uniform_pricing} does not rely discontinuously on the buyer being completely uninformed. 
Instead, the approximation degrades smoothly with the informativeness of the buyer's outside signal. 
For example, if the buyer learns his value perfectly with probability $\rho$ and is otherwise uninformed, the uniform-pricing guarantee becomes $\frac{1-\rho}{2}$ of the welfare benchmark. 
More generally, simple pricing remains competitive whenever the seller has a meaningful informational advantage over the buyer. 
The practical message is therefore not that the seller must be the buyer's only source of information, but that coarse disclosure can make simple posted pricing effective in environments where the seller has a meaningful informational advantage through control over product disclosure.

\paragraph{A Tight Example}
Next, we show that the approximation guarantee of half of the optimal revenue is tight for uniform pricing mechanisms by constructing a simple example with two products.
\begin{example}
There are two products for sale, and the buyer has two possible types drawn from the following distribution parameterized by $\epsilon > 0$:
\begin{equation*}
(v_1, v_2) =
\begin{cases}
  (0,1) & \text{with probability } 1-\epsilon, \\
  (\frac{1}{\epsilon},1) & \text{with probability } \epsilon.
\end{cases}
\end{equation*} 
\end{example}
In this example, the optimal mechanism is to fully disclose the value to the buyer, charge a price $\frac{1}{\epsilon}$ for the first product, and a price $1$ for the second product. 
The seller extracts the full surplus in this mechanism, with expected revenue equal to $2-\epsilon$. 

Now suppose the seller is restricted to uniform pricing mechanisms. 
In this case, if the price is at most 1, then the expected revenue given any information structure is at most~1. 
If the price is strictly larger than 1, given any information structure, the second product cannot be sold to the buyer, as the buyer has a fixed value of 1 on product $2$ for all types. 
In this case, the expected welfare from selling the product to the buyer is at most 1, which implies that the expected revenue is at most 1. 
Therefore, the optimal revenue from uniform pricing is at most~1, 
and the approximation ratio of uniform pricing in this example is $2-\epsilon$. 
By taking $\epsilon$ to 0, the worst-case gap is 2. 

The simple intuition behind the tight example is the presence of two distinct buyer types: one occurs with low probability but has a very high maximum value, while the other occurs with high probability but has a much lower maximum value. Importantly, both types contribute equally to the overall welfare. If the uniform price is set too low, the seller captures little revenue from the high-value type compared to its welfare. Conversely, if the price is set too high, the low-value type is entirely excluded, contributing nothing to welfare. As a result, any uniform pricing mechanism must trade off between these two extremes and cannot guarantee more than half of the optimal welfare. In contrast, pricing mechanisms extract the full surplus in this example. 

\subsection{Improved Approximation of Pricing Mechanisms}
A second price allows the seller to serve buyers with lower expected values without giving up as much revenue from buyers with higher expected values. With a suitable information policy, two price levels guarantee three-fifths of welfare for any number of products and any correlation among their values.

\begin{theorem}[Improved Approximation of Pricing Mechanisms]
\label{thm:approx_opt}
For any finite number of products and any distribution $\dist$ with $\optwel<\infty$, there exists a pricing mechanism $\mech$ with coarse horizontal disclosure and at most two distinct prices such that
\[
\rev(\mech)\geq\frac35\optwel\geq\frac35\optrev.
\] 
\end{theorem}
The proof starts from the best uniform price. If its revenue is already at least three-fifths of welfare, no price differentiation is needed. Otherwise, the limits on uniform-price revenue identify a set of products that can be sold cheaply while other products command a higher price. The information policy pools some low-value states with high-value states to support the expensive sales, checking that the buyer prefers an expensive product to every cheaper alternative. It then uses the remaining low-value states to support sales at the lower price. The revenue bound accounts for both uses of the same pool of states.

The three-fifths factor is a lower bound; the exact worst-case factor remains open. The proof in \cref{subapx:three_fifths} separates the pricing construction from the scalar inequality that determines the bound. The appendix also provides a self-contained verification of that inequality by exact rational arithmetic. 

\section{Discussions and Extensions}
\label{sec:discussion}

In the previous section, we showed that both uniform pricing and pricing mechanisms are competitive when the seller has control over the information structure. 
We also established that pricing mechanisms strictly outperform uniform pricing mechanisms in the worst-case scenario. 
In this section, we provide additional results that clarify the scope and interpretation of these guarantees. 

We first relax the assumption that the buyer is initially uninformed. 
This is important for applications: buyers often have access to outside information, but platforms may still be better informed about product attributes, match quality, or the relevant comparisons across products. 
We show that the guarantee of uniform pricing degrades smoothly with the informativeness of the buyer's outside signal. 
Thus, the main insight extends to environments where the seller has an informational advantage rather than being the buyer's only source of information.

We then establish exact revenue optimality with two products and arbitrary correlation. For any number of products, negatively affiliated or exchangeable values allow pricing mechanisms to extract the full surplus. 
Together, these results further illustrate the competitiveness and robustness of pricing mechanisms under endogenous information disclosure. 
All proofs in this section are provided in \cref{apx:proofs}.

\subsection{Partially Informed Buyer}
\label{subsec:partially_informed_buyer}
This subsection formalizes the intermediate information environment discussed above. 
The buyer may receive an outside signal before interacting with the seller, but the seller still controls the additional disclosure of product information.

Suppose that before the seller's disclosure, the buyer privately observes a signal $t\in\mathcal{T}$ generated by an exogenous experiment $\eta:\vals\to\Delta(\mathcal{T})$. 
The experiment $\eta$ is common knowledge. 
The seller still commits ex ante to an information structure $(\signals,\experiment)$ and a selling mechanism; the menu and prices do not directly condition on the realized initial signal $t$. 
Conditional on $\val$, the initial signal $t\sim\eta(\val)$ and the seller's signal $\signal\sim\experiment(\val)$ are drawn independently. 
Given $(t,\signal)$, let
\begin{align*}
\belief_i^\eta(t,\signal)
=
\expect[\val\sim\dist,\;t'\sim\eta(\val),\;\signal'\sim\experiment(\val)]
{\val_i \given t'=t,\signal'=\signal}
\end{align*}
be the buyer's posterior mean value for product $i$. 
For a mechanism $\mech$ in this extended environment, let $\rev_\eta(\mech)$ denote the expected revenue under the joint distribution of $\val\sim\dist$, $t\sim\eta(\val)$, and $\signal\sim\experiment(\val)$. 
The welfare benchmark remains $\optwel=\expect[\dist]{\max_{i\in[m]}\val_i}$.

We measure the informativeness of the buyer's initial signal through the following likelihood-ratio tail condition.

\begin{definition}[Approximately Bounded Initial Information]
\label{def:approximately_bounded_initial_information}
Fix $\dist$ and $\eta:\vals\to\Delta(\mathcal{T})$. 
For every signal $t$ with positive ex-ante probability under $\val\sim\dist$ and $t\sim\eta(\val)$, define
\begin{align*}
L_{\dist,\eta}(t)
=
\inf\left\{
K\geq 1:
\eta(t\mid \val)\leq K\eta(t\mid \val')
\text{ for all }\val,\val'\in\vals
\right\},
\end{align*}
with the convention that $L_{\dist,\eta}(t)=+\infty$ if no finite $K$ satisfies the inequalities above. 
For $\lambda\geq 0$, let
\begin{align*}
\beta_{\dist,\eta}(\lambda)
=
\prob[\val\sim\dist,\;t\sim\eta(\val)]
{
L_{\dist,\eta}(t)>e^\lambda
}.
\end{align*}
We say that $\eta$ is $(\lambda,\beta)$-approximately bounded relative to $\dist$ if $\beta_{\dist,\eta}(\lambda)\leq \beta$.
\end{definition}

The probability in $\beta_{\dist,\eta}(\lambda)$ is taken only over the value profile and the buyer's initial signal. 
Thus $\beta_{\dist,\eta}(\lambda)$ is the probability that the buyer receives an initial signal whose likelihood-ratio distortion is larger than $e^\lambda$.

\begin{proposition}[Uniform Pricing with a Partially Informed Buyer]
\label{prop:partially_informed_buyer}
Suppose the buyer initially observes $t\sim\eta(\val)$, and suppose $\eta$ is $(\lambda,\beta)$-approximately bounded relative to $\dist$. 
Then there exists a uniform pricing mechanism $\mech$ with a coarse horizontal disclosure such that
\begin{align*}
\rev_\eta(\mech)
\geq
\frac{e^{-\lambda}(1-\beta)}{2}\cdot\optwel.
\end{align*}
\end{proposition}

The proof uses the same coarse horizontal disclosure as in \cref{prop:uniform_pricing}. 
The seller lowers the uniform price by a factor of $e^{-\lambda}$, which guarantees a sale whenever the buyer's initial signal has likelihood-ratio distortion at most $e^\lambda$. 
The only remaining loss is the probability $\beta$ of signals that may be more informative.

This result nests several natural cases. 
If the buyer has no initial information, then $L_{\dist,\eta}(t)=1$ for every realized signal $t$, so $\eta$ is $(0,0)$-approximately bounded. 
Thus \cref{prop:partially_informed_buyer} recovers the half-welfare guarantee in \cref{prop:uniform_pricing}. 

If $L_{\dist,\eta}(t)\leq K$ for every realized signal $t$, then $\eta$ is $(\log K,0)$-approximately bounded, and the guarantee becomes
\begin{align*}
\rev_\eta(\mech)
\geq
\frac{1}{2K}\cdot\optwel.
\end{align*}

The condition also allows the buyer to learn his value profile perfectly with some probability. 
Suppose that, independently of $\val$, the buyer observes no information with probability $1-\rho$ and observes the full value profile $\val$ with probability $\rho$. 
That is,
\begin{align*}
t=
\begin{cases}
\perp, & \text{with probability }1-\rho,\\
\val, & \text{with probability }\rho.
\end{cases}
\end{align*}
Then $L_{\dist,\eta}(\perp)=1$, while every perfectly revealing signal has infinite likelihood-ratio distortion. 
Hence $\eta$ is $(0,\rho)$-approximately bounded, and
\begin{align*}
\rev_\eta(\mech)
\geq
\frac{1-\rho}{2}\cdot\optwel.
\end{align*}
More generally, if with probability $1-\rho$ the buyer receives a $K$-bounded signal and with probability $\rho$ he learns $\val$ perfectly, then
\begin{align*}
\rev_\eta(\mech)
\geq
\frac{1-\rho}{2K}\cdot\optwel.
\end{align*}

\subsection{Optimal Pricing with Two Products}
\label{subsec:twoitems}

With two products, the seller can achieve optimal revenue using deterministic prices. This conclusion allows arbitrary correlation between values and optimizes the information policy together with the prices.

\begin{theorem}[Optimal Pricing with Two Products]\label{thm:two_item_optimal}
For any distribution $\dist$ over values for two products with $\optwel<\infty$, there exists a pricing mechanism $\mech$ and an information structure with at most two signals, each recommending a purchase, such that
\[
 \rev(\mech)=\optrev.
\]
\end{theorem}
The proof uses pooling to replace partial allocations with certain sales. Starting from a hypothetical optimum with lotteries, it pools their signals with the signal recommending the cheaper product. Incentive constraints at zero-rent signals bound the revenue from partial allocations, and these bounds show that the gain from certain sales exceeds the price adjustment needed to preserve incentives. The proof is in \cref{subapx:two_item_optimal}.

The information policy in \cref{thm:two_item_optimal} may depend on the full value profile. Under the more restrictive coarse horizontal disclosure, pricing still guarantees two-thirds of welfare. This welfare bound is tight even when the seller can use arbitrary mechanisms and information structures; exact revenue optimality therefore does not imply full surplus extraction.

\begin{proposition}\label{prop:two_item_approx}
For any distribution over valuations for two products, there exists a pricing mechanism with coarse horizontal disclosure that achieves at least a \(\frac{2}{3}\) fraction of the optimal welfare. Moreover, this bound is tight.
\end{proposition}
For the welfare guarantee, it suffices to compare no-information sale of each product with horizontal disclosure and suitable prices. The revenues of these three mechanisms sum to at least twice optimal welfare. The tight example in \cref{subapx:two_item_approx} shows that the remaining loss reflects the limits of revenue extraction, even with optimal information and mechanisms. 

\subsection{Full Surplus Extraction}
\label{subsec:fullsurplus}
\label{sub:full_surplus_multi}

In this section, we provide a sufficient condition on the correlation structures of the value distributions such that, when the seller has joint design power over information structures and selling mechanisms, pricing mechanisms are revenue-optimal for the seller. Moreover, we show that the optimal information structures in this case can be captured by horizontal disclosure (see \cref{def:coarse_horizontal}), and the optimal mechanism extracts the full surplus. The following lemma demonstrates the sufficient condition.
In \cref{subapx:full_surplus_multi}, we show that this sufficient condition is also necessary under a mild generic condition. 

\begin{lemma}\label{lem:full_surplus}
For any number of products and for any distribution $\dist$, 
there exists a pricing mechanism $\mech$ with horizontal disclosure that extracts the full surplus, 
i.e., $\rev(\mech) = \optwel$,
if for any $i\neq i'$, we have 
\begin{align*}
\expect[\dist]{ \val_i \given i^*(\val) = i} 
\geq \expect[\dist]{ \val_i \given i^*(\val) = i'}.
\end{align*}
\end{lemma}
Essentially, \cref{lem:full_surplus} shows that full surplus extraction under pricing mechanisms is feasible if, for each product $i$, the posterior mean of product $i$ is maximal conditional on the signal indicating that $i$ has the highest realized value among all products. When this holds, posting a price equal to the posterior mean under that signal for each product $i$ automatically satisfies incentive compatibility and extracts the full surplus. 
This condition does not hold in general; for example, it may fail under asymmetric distributions with positive correlation. We provide a concrete counterexample in \cref{subapx:no_full_surplus}, where the condition fails and full surplus extraction is impossible.

To better understand the results, in this section, we provide two additional sufficient conditions on the primitives such that the conditions in \cref{lem:full_surplus} are satisfied. 
We first show that it is satisfied for a general family of negatively correlated distributions. 

\begin{definition}[Negative Affiliation]
\label{def:negative_affiliation}
A distribution $\dist$ is \emph{negatively affiliated} if its density function $\density$ satisfies the log-submodularity property. That is, for any pair of value profiles $\val,\val'\in\reals^m_+$, 
we have 
\begin{align*}
\density(\val) \cdot \density(\val')
\geq \density(\val\wedge\val') \cdot \density(\val\vee\val'), 
\end{align*}
where $\val\wedge\val'$ is the component-wise minimum of $\val$ and $\val'$, and $\val\vee\val'$ is the component-wise maximum.
\end{definition}
The negative affiliation condition is the opposite of the positive affiliation condition in \citet{milgrom1982theory}. It assumes a form of weakly negative correlation that satisfies the well-known Fortuin–Kasteleyn–Ginibre (FKG) lattice condition~\citep{fortuin1971correlation} and includes independent values as special cases.\footnote{The standard FKG inequality only applies for log-supermodular distributions. As we will show in \cref{subapx:full_surplus_multi}, we can redefine the meet and join operations such that the distribution is log-supermodular in the new lattice space, where the FKG inequality is applicable.} 
Intuitively, with negative correlation, the posterior mean value of product $i$ should be higher when product $i$ has the highest value compared to when product $j$ has the highest value. 

\begin{proposition}[Full Surplus Extraction under Negative Affiliation]
\label{thm:full_surplus_multi}
For any number of products and for any distribution~$\dist$ that is negatively affiliated, 
there exists a pricing mechanism~$\mech$ with horizontal disclosure that extracts the full surplus, 
i.e., $\rev(\mech) = \optwel = \optrev$. 
\end{proposition}

In our model, in addition to negatively correlated distributions, it is also possible for some positively correlated distributions to satisfy the conditions in \cref{lem:full_surplus}. 
Our next sufficient condition is exchangeability, which includes positively correlated distributions. 

\begin{definition}[Exchangeable distribution]
Let \(X=(X_1,\ldots,X_m)\sim \dist\) be a random valuation vector over \(m\) items. 
We say that \(\dist\) is \emph{exchangeable} if it is invariant under coordinate permutations: for every permutation \(\pi\),
\[
(X_1,\ldots,X_m)\stackrel{d}{=}(X_{\pi(1)},\ldots,X_{\pi(m)}).
\]
\end{definition}
This condition also appears in \citet{deb2021multi}. It essentially assumes a symmetry across products. 
Under exchangeability, the posterior mean value of each product $i$, conditioning on it being the highest value product, is the same for all products, and is higher than conditioning on any other signal. 

\begin{proposition}[Full Surplus Extraction under Exchangeability]
\label{thm:full_surplus_exchange}
For any number of products and for any distribution~$\dist$ that is exchangeable, 
there exists a {uniform} pricing mechanism~$\mech$ with horizontal disclosure that extracts the full surplus, 
i.e., $\rev(\mech) = \optwel = \optrev$. 
\end{proposition}

\section{Conclusions}
\label{sec:conclusion}
In this paper, we show that, unlike in classic multi-product settings where simple mechanisms fail to provide any approximation guarantees for correlated distributions, the additional power of information design enables simple deterministic pricing mechanisms to achieve a constant-factor approximation to the optimal revenue. 
Uniform pricing guarantees half of welfare, and allowing two prices raises the guarantee to three-fifths. With two products, deterministic pricing attains optimal revenue for every joint distribution. The seller achieves these results by shaping what the buyer learns before choosing from the menu. 
We also show that this insight extends beyond the benchmark of an initially uninformed buyer: when the buyer has access to outside information, the guarantee of uniform pricing degrades smoothly with the informativeness of that outside signal.
Our paper also leaves several important open questions that warrant further investigation.
First, the exact welfare guarantee of pricing mechanisms for a general number of products remains open between our lower bound of $\frac35$ and the upper bound of $\frac23$ implied by the two-product example. We conjecture that $\frac23$ is the exact guarantee. 
The two-product optimality result and the three-product counterexample also identify a boundary for the role of lotteries. Understanding which distributions make random allocations valuable with three or more products remains an open question.
Finally, it would be interesting to extend our results to other environments, such as settings where the buyer has more complex combinatorial valuations.

\bibliographystyle{apalike}
\bibliography{references}

@article{akbarpour2023algorithmic,
  author = {Akbarpour, Mohammad and Kominers, Scott Duke and Li, Kevin Michael and Li, Shengwu and Milgrom, Paul},
  title = {Algorithmic mechanism design with investment},
  journal = {Econometrica},
  volume = {91},
  number = {6},
  pages = {1969--2003},
  year = {2023},
  doi = {10.3982/ecta19559},
}

@article{armstrong_multiproduct_1996,
  author = {Armstrong, Mark},
  title = {Multiproduct {Nonlinear} {Pricing}},
  journal = {Econometrica},
  volume = {64},
  number = {1},
  pages = {51--75},
  year = {1996},
  doi = {10.2307/2171924},
}

@article{athey2011position,
  author = {Athey, Susan and Ellison, Glenn},
  title = {Position auctions with consumer search},
  journal = {The Quarterly Journal of Economics},
  volume = {126},
  number = {3},
  pages = {1213--1270},
  year = {2011},
  doi = {10.1093/qje/qjr028},
}

@inproceedings{babaioff2017menu,
  author = {Babaioff, Moshe and Gonczarowski, Yannai A. and Nisan, Noam},
  title = {The menu-size complexity of revenue approximation},
  booktitle = {Proceedings of the 49th Annual ACM SIGACT Symposium on Theory of Computing},
  pages = {869--877},
  year = {2017},
  publisher = {ACM},
  doi = {10.1145/3055399.3055426},
}

@article{babaioff2020simple,
  author = {Babaioff, Moshe and Immorlica, Nicole and Lucier, Brendan and Weinberg, S. Matthew},
  title = {A simple and approximately optimal mechanism for an additive buyer},
  journal = {Journal of the ACM},
  volume = {67},
  number = {4},
  pages = {24:1--24:40},
  year = {2020},
  doi = {10.1145/3398745},
}

@article{bergemann2007information,
  author = {Bergemann, Dirk and Pesendorfer, Martin},
  title = {Information structures in optimal auctions},
  journal = {Journal of Economic Theory},
  volume = {137},
  number = {1},
  pages = {580--609},
  year = {2007},
  doi = {10.1016/j.jet.2007.02.001},
}

@article{Bergemann2022Optimal,
  author = {Bergemann, Dirk and Heumann, Tibor and Morris, Stephen and Sorokin, Constantine and Winter, Eyal},
  title = {Optimal Information Disclosure in Classic Auctions},
  journal = {American Economic Review: Insights},
  volume = {4},
  number = {3},
  pages = {371--388},
  year = {2022},
  doi = {10.1257/aeri.20210504},
}

@article{bergemann2022screening,
  author = {Bergemann, Dirk and Heumann, Tibor and Morris, Stephen},
  title = {Screening with persuasion},
  journal = {Journal of Political Economy},
  volume = {134},
  number = {2},
  pages = {570--625},
  year = {2026},
  doi = {10.1086/738342},
}

@article{bergemann2022third,
  author = {Bergemann, Dirk and Castro, Francisco and Weintraub, Gabriel},
  title = {Third-degree price discrimination versus uniform pricing},
  journal = {Games and Economic Behavior},
  volume = {131},
  pages = {275--291},
  year = {2022},
  doi = {10.1016/j.geb.2021.11.008},
}

@article{bergemann2025alignment,
  author = {Bergemann, Dirk and Brooks, Benjamin and Morris, Stephen},
  title = {On the Alignment of Consumer Surplus and Total Surplus Under Competitive Price Discrimination},
  journal = {American Economic Journal: Microeconomics},
  volume = {17},
  number = {4},
  pages = {234--259},
  year = {2025},
  doi = {10.1257/mic.20240197},
}

@article{bikhchandani2022selling,
  author = {Bikhchandani, Sushil and Mishra, Debasis},
  title = {Selling two identical objects},
  journal = {Journal of Economic Theory},
  volume = {200},
  pages = {105397},
  year = {2022},
  doi = {10.1016/j.jet.2021.105397},
}

@article{bulow1989simple,
  author = {Bulow, Jeremy and Roberts, John},
  title = {The simple economics of optimal auctions},
  journal = {Journal of Political Economy},
  volume = {97},
  number = {5},
  pages = {1060--1090},
  year = {1989},
  doi = {10.1086/261643},
}

@inproceedings{cai2016duality,
  author = {Cai, Yang and Devanur, Nikhil R. and Weinberg, S. Matthew},
  title = {A duality-based unified approach to {Bayesian} mechanism design},
  booktitle = {Proceedings of the forty-eighth annual ACM symposium on Theory of Computing},
  pages = {926--939},
  year = {2016},
  publisher = {ACM},
  doi = {10.1145/2897518.2897645},
}

@inproceedings{cai2024algorithmic,
  author = {Cai, Yang and Li, Yingkai and Wu, Jinzhao},
  title = {Algorithmic Information Disclosure in Optimal Auctions},
  booktitle = {Proceedings of the 25th ACM Conference on Economics and Computation},
  pages = {43},
  year = {2024},
  publisher = {ACM},
  doi = {10.1145/3670865.3673566},
}

@inproceedings{cai_simple_2017,
  author = {Cai, Yang and Zhao, Mingfei},
  title = {Simple mechanisms for subadditive buyers via duality},
  booktitle = {Proceedings of the 49th Annual ACM SIGACT Symposium on Theory of Computing},
  pages = {170--183},
  year = {2017},
  publisher = {ACM},
  doi = {10.1145/3055399.3055465},
}

@inproceedings{cai_simultaneous_2023,
  author = {Cai, Yang and Chen, Ziyun and Wu, Jinzhao},
  title = {Simultaneous auctions are approximately revenue-optimal for subadditive bidders},
  booktitle = {2023 IEEE 64th Annual Symposium on Foundations of Computer Science (FOCS)},
  pages = {134--147},
  year = {2023},
  publisher = {IEEE},
  doi = {10.1109/focs57990.2023.00017},
}

@article{carroll2017robustness,
  author = {Carroll, Gabriel},
  title = {Robustness and separation in multidimensional screening},
  journal = {Econometrica},
  volume = {85},
  number = {2},
  pages = {453--488},
  year = {2017},
  doi = {10.3982/ecta14165},
}

@inproceedings{chawla2010multi,
  author = {Chawla, Shuchi and Hartline, Jason D. and Malec, David L. and Sivan, Balasubramanian},
  title = {Multi-parameter mechanism design and sequential posted pricing},
  booktitle = {Proceedings of the forty-second ACM symposium on Theory of computing},
  pages = {311--320},
  year = {2010},
  publisher = {ACM},
  doi = {10.1145/1806689.1806733},
}

@article{chen2023information,
  author = {Chen, Yi-Chun and Yang, Xiangqian},
  title = {Information design in optimal auctions},
  journal = {Journal of Economic Theory},
  volume = {212},
  pages = {105710},
  year = {2023},
  doi = {10.1016/j.jet.2023.105710},
}

@article{daskalakis2017strong,
  author = {Daskalakis, Constantinos and Deckelbaum, Alan and Tzamos, Christos},
  title = {Strong Duality for a Multiple-Good Monopolist},
  journal = {Econometrica},
  volume = {85},
  number = {3},
  pages = {735--767},
  year = {2017},
  doi = {10.3982/ecta12618},
}

@article{daskalakis_multi-item_2022,
  author = {Daskalakis, Constantinos and Fishelson, Maxwell and Lucier, Brendan and Syrgkanis, Vasilis and Velusamy, Santhoshini},
  title = {Multi-{Item} {Nontruthful} {Auctions} {Achieve} {Good} {Revenue}},
  journal = {SIAM Journal on Computing},
  volume = {51},
  number = {6},
  pages = {1796--1838},
  year = {2022},
  doi = {10.1137/22m1471742},
}

@article{deb2021multi,
  author = {Deb, Rahul and Roesler, Anne-Katrin},
  title = {Multi-Dimensional Screening: Buyer-Optimal Learning and Informational Robustness},
  journal = {The Review of Economic Studies},
  volume = {91},
  number = {5},
  pages = {2744--2770},
  year = {2024},
  doi = {10.1093/restud/rdad100},
}

@article{edelman2007internet,
  author = {Edelman, Benjamin and Ostrovsky, Michael and Schwarz, Michael},
  title = {Internet advertising and the generalized second-price auction: Selling billions of dollars worth of keywords},
  journal = {American Economic Review},
  volume = {97},
  number = {1},
  pages = {242--259},
  year = {2007},
  doi = {10.1257/aer.97.1.242},
}

@article{esHo2007optimal,
  author = {Es{\H{o}}, P{\'e}ter and Szentes, Bal{\'a}zs},
  title = {Optimal information disclosure in auctions and the handicap auction},
  journal = {The Review of Economic Studies},
  volume = {74},
  number = {3},
  pages = {705--731},
  year = {2007},
  doi = {10.1111/j.1467-937x.2007.00442.x},
}

@inproceedings{feng2023simple,
  author = {Feng, Yiding and Hartline, Jason D. and Li, Yingkai},
  title = {Simple mechanisms for non-linear agents},
  booktitle = {Proceedings of the 2023 Annual ACM-SIAM Symposium on Discrete Algorithms (SODA)},
  pages = {3802--3816},
  year = {2023},
  publisher = {Society for Industrial and Applied Mathematics},
  doi = {10.1137/1.9781611977554.ch148},
}

@article{fortuin1971correlation,
  author = {Fortuin, Cees M. and Kasteleyn, Pieter W. and Ginibre, Jean},
  title = {Correlation inequalities on some partially ordered sets},
  journal = {Communications in Mathematical Physics},
  volume = {22},
  number = {2},
  pages = {89--103},
  year = {1971},
  doi = {10.1007/bf01651330},
}

@article{haghpanah2021pure,
  author = {Haghpanah, Nima and Hartline, Jason},
  title = {When is pure bundling optimal?},
  journal = {The Review of Economic Studies},
  volume = {88},
  number = {3},
  pages = {1127--1156},
  year = {2021},
  doi = {10.1093/restud/rdaa053},
}

@article{hart2017approximate,
  author = {Hart, Sergiu and Nisan, Noam},
  title = {Approximate revenue maximization with multiple items},
  journal = {Journal of Economic Theory},
  volume = {172},
  pages = {313--347},
  year = {2017},
  doi = {10.1016/j.jet.2017.09.001},
}

@article{hart2019selling,
  author = {Hart, Sergiu and Nisan, Noam},
  title = {Selling multiple correlated goods: Revenue maximization and menu-size complexity},
  journal = {Journal of Economic Theory},
  volume = {183},
  pages = {991--1029},
  year = {2019},
  doi = {10.1016/j.jet.2019.07.006},
}

@inproceedings{hartline2009simple,
  author = {Hartline, Jason D. and Roughgarden, Tim},
  title = {Simple versus optimal mechanisms},
  booktitle = {Proceedings of the 10th ACM conference on Electronic commerce},
  pages = {225--234},
  year = {2009},
  publisher = {ACM},
  doi = {10.1145/1566374.1566407},
}

@article{hartline2012approximation,
  author = {Hartline, Jason D.},
  title = {Approximation in mechanism design},
  journal = {American Economic Review},
  volume = {102},
  number = {3},
  pages = {330--336},
  year = {2012},
  doi = {10.1257/aer.102.3.330},
}

@article{hartline2015non,
  author = {Hartline, Jason D. and Lucier, Brendan},
  title = {Non-optimal mechanism design},
  journal = {American Economic Review},
  volume = {105},
  number = {10},
  pages = {3102--3124},
  year = {2015},
  doi = {10.1257/aer.20130712},
}

@inproceedings{jin2019tight,
  author = {Jin, Yaonan and Lu, Pinyan and Qi, Qi and Tang, Zhihao Gavin and Xiao, Tao},
  title = {Tight approximation ratio of anonymous pricing},
  booktitle = {Proceedings of the 51st Annual ACM SIGACT Symposium on Theory of Computing},
  pages = {674--685},
  year = {2019},
  publisher = {ACM},
  doi = {10.1145/3313276.3316331},
}

@article{kamenica2011bayesian,
  author = {Kamenica, Emir and Gentzkow, Matthew},
  title = {Bayesian persuasion},
  journal = {American Economic Review},
  volume = {101},
  number = {6},
  pages = {2590--2615},
  year = {2011},
  doi = {10.1257/aer.101.6.2590},
}

@article{kleiner2024extreme,
  author = {Kleiner, Andreas and Moldovanu, Benny and Strack, Philipp and Whitmeyer, Mark},
  title = {The extreme points of mean-preserving contractions},
  journal = {Journal of Economic Theory},
  volume = {237},
  pages = {106251},
  year = {2026},
  doi = {10.1016/j.jet.2026.106251},
}

@article{li2017discriminatory,
  author = {Li, Hao and Shi, Xianwen},
  title = {Discriminatory information disclosure},
  journal = {American Economic Review},
  volume = {107},
  number = {11},
  pages = {3363--3385},
  year = {2017},
  doi = {10.1257/aer.20151743},
}

@inproceedings{li2022selling,
  author = {Li, Yingkai},
  title = {Selling data to an agent with endogenous information},
  booktitle = {Proceedings of the 23rd ACM Conference on Economics and Computation},
  pages = {664--665},
  year = {2022},
  publisher = {ACM},
  doi = {10.1145/3490486.3538360},
}

@article{manelli2006bundling,
  author = {Manelli, Alejandro M. and Vincent, Daniel R.},
  title = {Bundling as an optimal selling mechanism for a multiple-good monopolist},
  journal = {Journal of Economic Theory},
  volume = {127},
  number = {1},
  pages = {1--35},
  year = {2006},
  doi = {10.1016/j.jet.2005.08.007},
}

@article{milgrom1982theory,
  author = {Paul R. Milgrom and Robert J. Weber},
  title = {A theory of auctions and competitive bidding},
  journal = {Econometrica},
  volume = {50},
  number = {5},
  pages = {1089--1122},
  year = {1982},
  doi = {10.2307/1911865},
}

@article{mirrokni2020non,
  author = {Mirrokni, Vahab and Paes Leme, Renato and Tang, Pingzhong and Zuo, Song},
  title = {Non-Clairvoyant Dynamic Mechanism Design},
  journal = {Econometrica},
  volume = {88},
  number = {5},
  pages = {1939--1963},
  year = {2020},
  doi = {10.3982/ecta15530},
}

@article{mussa1978monopoly,
  author = {Mussa, Michael and Rosen, Sherwin},
  title = {Monopoly and product quality},
  journal = {Journal of Economic Theory},
  volume = {18},
  number = {2},
  pages = {301--317},
  year = {1978},
  doi = {10.1016/0022-0531(78)90085-6},
}

@inproceedings{ravid2024monopoly,
  author = {Ravid, Doron and Mensch, Jeffrey},
  title = {Monopoly, Product Quality, and Flexible Learning},
  booktitle = {Proceedings of the 25th ACM Conference on Economics and Computation},
  pages = {855},
  year = {2024},
  publisher = {ACM},
  doi = {10.1145/3670865.3673608},
}

@article{rochet_ironing_1998,
  author = {Rochet, Jean-Charles and Chon{\'e}, Philippe},
  title = {Ironing, {Sweeping}, and {Multidimensional} {Screening}},
  journal = {Econometrica},
  volume = {66},
  number = {4},
  pages = {783--826},
  year = {1998},
  doi = {10.2307/2999574},
}

@article{roughgarden2019approximately,
  author = {Roughgarden, Tim and Talgam-Cohen, Inbal},
  title = {Approximately optimal mechanism design},
  journal = {Annual Review of Economics},
  volume = {11},
  number = {1},
  pages = {355--381},
  year = {2019},
  doi = {10.1146/annurev-economics-080218-025607},
}

@article{schottmuller2023optimal,
  author = {Schottm{\"u}ller, Christoph},
  title = {Optimal information structures in bilateral trade},
  journal = {Theoretical Economics},
  volume = {18},
  number = {1},
  pages = {421--461},
  year = {2023},
  doi = {10.3982/te4074},
}

@article{shi2024multi,
  author = {Shi, Fanqi and Zhang, Yutong and Chen, Litian},
  title = {Multi-product Monopolist and Information Design},
  journal = {The RAND Journal of Economics},
  volume = {57},
  number = {1},
  pages = {16--38},
  year = {2026},
  doi = {10.1111/1756-2171.70033},
}

@article{thirumulanathan2019optimal,
  author = {Thirumulanathan, D. and Sundaresan, Rajesh and Narahari, Y.},
  title = {On optimal mechanisms in the two-item single-buyer unit-demand setting},
  journal = {Journal of Mathematical Economics},
  volume = {82},
  pages = {31--60},
  year = {2019},
  doi = {10.1016/j.jmateco.2019.01.005},
}

@misc{wang2026replacing,
  author = {Wang, Pengbo},
  title = {Replacing Product Lotteries with Information Pooling},
  year = {2026},
  howpublished = {Working paper, Zenodo},
  doi = {10.5281/zenodo.24512199},
  url = {https://zenodo.org/records/24512199},
}

@article{wei2024reverse,
  author = {Wei, Dong and Green, Brett},
  title = {{(Reverse)} Price Discrimination with Information Design},
  journal = {American Economic Journal: Microeconomics},
  volume = {16},
  number = {2},
  pages = {267--295},
  year = {2024},
  doi = {10.1257/mic.20220242},
}
\newpage
\appendix

\section{Missing Proofs}
\label{apx:proofs}
\subsection{Approximation Without Information Design}
\label{subapx:without_info}
\begin{proposition}[Proposition A.12 of \citet{hart2019selling}]
\label{prop:without_info}
For any $\epsilon>0$, even with two products, there exists a distribution over values and a fully revealing information structure of the buyer such that the expected revenue from pricing mechanisms is less than $\epsilon$ fraction of the optimal revenue. 
\end{proposition}

\subsection{Approximation of Uniform Pricing Mechanisms}
\begin{proof}[Proof of \cref{prop:uniform_pricing}]
For any uniform price $\pay>0$, let
\begin{align*}
I_+(\pay) = \lbr{i\in[m] : \expect[\dist]{\val_i \given i=i^*(\val)} \geq \pay}.
\end{align*}
That is, by only disclosing to the buyer on which product he values the most, $I_+(\pay)$ is the set of products such that the posterior value of the buyer on the highest value product is at least $\pay$. 
Let $I_-(\pay) = [m]\backslash I_+(\pay)$. 
Moreover, let 
\begin{align*}
\wel_+(\pay) = \expect[\dist]{\max_i \val_i \cdot\indicate{i^*(\val)\in I_+(\pay)}}
\quad\text{and}\quad
\wel_-(\pay) = \expect[\dist]{\max_i \val_i \cdot\indicate{i^*(\val)\in I_-(\pay)}}.
\end{align*}

{
According to \cref{def:coarse_horizontal}, the signal space
\(\signals_{\pay} = 2^{[m]} \setminus \{\emptyset\}\),
together with any mapping $\tau \colon [m] \to \Delta(\signals_{\pay})$,
induces a coarse horizontal disclosure information structure.
For a given uniform price $\pay>0$, we focus on a particular family of such
coarse horizontal disclosure structures by imposing the following constraints
on the mapping $\tau$:
}
\begin{enumerate}
\item for any product $i\in I_+(\pay)$, $\tau(i)$ is deterministic, and for any $j\in \tau(i)$ and $j\neq i$, we have $j\not\in I_+(\pay)$;
\item \label{bullet:constraint} {for any product $i\in I_+(\pay)$, the posterior value on product $i$ for receiving signal $\tau(i)$ is at least~$\pay$. That is to say, define \(\experiment_{\pay}\) as the coarse horizontal information structure induced by \(\tau\), the following holds:
\begin{align}
\expect[\val \sim \dist, \signal \sim \experiment_{\pay}]{\val_i \given \signal = \tau(i)} \geq \pay.\label{eq:constraint_posterior}
\end{align}}
\end{enumerate}
Essentially, the information structure constructed above randomly pools types whose highest value lies below $\pay$ (i.e., whose highest-value product is in $I_-(\pay)$) with types whose highest value lies above $\pay$ (i.e., in $I_+(\pay)$), subject to the requirement that, after pooling, the buyer's posterior mean for every product $i \in I_+(\pay)$ remains weakly above $\pay$ upon receiving the signal $\tau(i)$ in the coarse horizontal disclosure. This ensures that, for all types with highest value above $\pay$, at least one product is sold at price $\pay$ after coarse horizontal disclosure.

Note that this description does not uniquely determine the coarse horizontal disclosure policy. Nevertheless, to maximize the probability that at least one product is sold, the coarse horizontal disclosure should pool below-$\pay$ types with above-$\pay$ types as much as possible, subject to the posterior constraint in \Cref{eq:constraint_posterior}.

Let $\pay^*$ be the {supremum} price such that there exists a coarse horizontal disclosure information structure $(\signals_{\pay^*},\experiment_{\pay^*})$ defined above under which a product is sold to the buyer with probability one. { It follows that the revenue from uniform pricing should be at least \(\pay^* - \epsilon\)}.
For any $\epsilon > 0$, we first show that for $(\signals_{\pay^*+\epsilon},\experiment_{\pay^*+\epsilon})$
that maximizes the selling probability, which is strictly less than 1 by the definition of $\pay^*$, we have 
\begin{align*}
\wel_+(\pay^*+\epsilon) 
&= \expect[\dist]{\sum_{i\in I_+(\pay^*+\epsilon)} \val_i \cdot\indicate{i^*(\val)=i}}\\
&\leq \expect[\experiment_{\pay^*+\epsilon},\dist]{\sum_{i\in I_+(\pay^*+\epsilon)} (\pay^*+\epsilon) \cdot\indicate{i\in\signal}}
\leq \pay^*+\epsilon.
\end{align*}
The first inequality holds since for any product $i \in I_+(\pay^*+\epsilon)$, the posterior value on product $i$ for receiving signal {$\tau(i)$} is exactly $\pay^*+\epsilon$. This is because otherwise we can further pool types with highest value product in $I_-(\pay^*+\epsilon)$ with those with highest value product in $I_+(\pay^*+\epsilon)$ to strictly increase the selling probability. 
Moreover, it is easy to show that 
\begin{align*}
\wel_-(\pay^*+\epsilon) = \expect[\dist]{\max_i \val_i \cdot\indicate{i^*(\val)\in I_-(\pay^*+\epsilon)}} \leq \pay^*+\epsilon.
\end{align*}

Let $\mech$ be the uniform pricing mechanism that sells the product using uniform price~{$\pay^* - \epsilon$}
in which the information structure $(\signals_{\pay^*},\experiment_{\pay^*})$ is chosen as a coarse horizontal disclosure information structure such that a product is sold with probability one. 
For any $\epsilon > 0$, we have
\begin{align*}
\optwel = \wel_+(\pay^*+\epsilon) + \wel_-(\pay^*+\epsilon) 
\leq 2(\pay^* + \epsilon)
= 2(\rev(\mech) + 2\epsilon).
\end{align*}
Since the above inequality holds for any $\epsilon > 0$, 
we have 
\begin{equation*}
\rev(\mech) \geq \frac{1}{2}\optwel \geq \frac{1}{2}\optrev.
\end{equation*}
Finally, the tightness of the approximation is shown in the tight example in \cref{sec:approximation}.
\end{proof}

\subsection{Approximation of Pricing Mechanisms}
\label{subapx:three_fifths}

\begin{proof}[Proof of \cref{thm:approx_opt}]
Write $W=\optwel$. We prove that, whenever $W>0$, either uniform pricing or a two-price mechanism earns strictly more than $3W/5$. The construction uses high-value states to support sales at the higher price while preserving low-value states that can support sales at the lower price. If both uniform pricing and this construction fall short, the resulting bounds on these two groups contradict a scalar inequality, stated and verified after the proof.

\paragraph{1. Pooling and normalization.}
Let $i^*(v)$ be the buyer's favorite product, breaking ties by index, and replace each favorite-product class of positive probability by its conditional mean
\[
 v^i=\mathbb E[v\mid i^*(v)=i].
\]
Product $i$ is still a favorite at $v^i$; we use $i$ to label this reduced state. Thus
\[
 \mathbb E\bigl[\max_j v^{i^*(v)}_j\bigr]
 =\mathbb E\bigl[v_{i^*(v)}\bigr]=W.
\]
An experiment on these at most $m$ states lifts to the original prior by disclosing only its final signal. Its posterior means and revenue are unchanged. We work with this finite prior and verify the set-valued signal requirement in \cref{def:coarse_horizontal} at the end.

Let $U$ be the supremum of common prices that an experiment on the reduced prior can make affordable after every signal. Since $W>0$, $0<U\le W$. A small discount implements any feasible price regardless of ties.

Measure values, prices, and welfare in units of $U$. There exist coefficients $a_j\in[0,1]$ and a nonnegative state variable $\eta$ such that
\begin{equation}\label{eq:tf_envelope}
 U=1,\qquad \mathbb E\eta=1,\qquad
 v_j\le a_j+(1-a_j)\eta\quad\text{for every }j.
\end{equation}
Here is a separation argument for this bound. With state probabilities $\mu_i$, a full-sale price $p$ is feasible precisely when some $x_{ij}\ge0$ satisfies
\[
 \sum_jx_{ij}=\mu_i,\qquad
 \sum_ix_{ij}(v^i_j-p)\ge0\quad\text{for each }j.
\]
Let $J=\{j:\max_i v^i_j>1\}$. If $J$ is empty, take $a_j=\eta_i=1$. Otherwise, for $p>1$, separation of the attainable surplus vectors from the nonnegative orthant gives weights $\lambda_j(p)>0$, normalized to sum to one, with
\[
 \sum_i\mu_i\max_{j\in J}\lambda_j(p)(v^i_j-p)<0.
\]
Take a subsequential limit as $p\downarrow1$. Every limiting weight is positive: a zero weight would make the maximum nonnegative in every state, whereas a positive weight on any $j\in J$ makes it strictly positive somewhere. Set $\Lambda_j=\lambda_j/\min_{k\in J}\lambda_k$, $a_j=1-\Lambda_j^{-1}$, and
\[
 \eta_i=1+\max_{j\in J}\Lambda_j(v^i_j-1).
\]
Then $\eta_i\ge0$, $\mathbb E\eta\le1$, and the desired bound holds. Give products outside $J$ coefficient $a_j=1$, and increase every $\eta_i$ by $1-\mathbb E\eta$. This proves \eqref{eq:tf_envelope}.

\paragraph{2. A two-price construction.}
Write $w=\max_jv_j$, $\mathcal L=\{\eta\le1\}$, and $\mathcal H=\{\eta>1\}$. Define
\[
 c=\Pr(\mathcal L),\quad \ell=\mathbb E[\eta\mathbf1_{\mathcal L}],\quad
 \Delta=\eta-w\text{ on }\mathcal H,\quad
 D=\mathbb E[\Delta\mathbf1_{\mathcal H}].
\]
The envelope implies $w\le1$ on $\mathcal L$ and $w\le\eta$ on $\mathcal H$. Hence
\[
 W=1+V,\qquad V=\mathbb E[(w-\eta)\mathbf1_{\mathcal L}]-D.
\]
If $V<2/3$, uniform pricing already gives the required strict inequality. Otherwise, suppose that uniform pricing and all the following two-price experiments earn at most
\begin{equation}\label{eq:tf_contradiction}
 R=\frac35(1+V),\qquad
 1\le R<\frac65,\qquad
 V=\frac53R-1\ge\frac23R.
\end{equation}
Both $\mathcal L$ and $\mathcal H$ have positive probability, so $0<c<1$; also $V\le\mathbb E[w\mathbf1_{\mathcal L}]\le c<1$, which justifies the bounds on $R$.

For $0<b\le1$, let
\[
 C_b=\{j:a_j\ge b\},\qquad
 L_b=\left\{\eta<1:\frac{w-\eta}{1-\eta}\ge b\right\},\qquad
 T_b=\mathbb E[w\mathbf1_{L_b}].
\]
Every state in $L_b$ has a favorite in $C_b$ worth at least $b$. Thus $T_b$ measures the value that can support sales at the lower price. Since $V+D>0$, $T_b/b\to\infty$ as $b\downarrow0$. Define its inverse profile and endpoint by
\begin{equation}\label{eq:tf_inverse}
 b(y)=\sup\{b\in(0,1]:T_b/b\ge y\}\quad(0<y\le c),
 \qquad \beta=b(c)>0.
\end{equation}
The profile is nonincreasing. Closed tails make each inverse cut feasible, including $T_\beta\ge\beta c$.

Fix the auxiliary variable $Y=\eta-\Delta/\beta$ on $\mathcal H$, and put
\[
 A(x)=\mathbb E[\mathbf1_{\mathcal H}(xY-1)_+],\qquad
 \dot A(x)=A'_+(x)=\mathbb E[Y\mathbf1_{\mathcal H\cap\{xY\ge1\}}].
\]
Although $Y$ can be negative, only its positive tail enters these expressions. We claim that, for every $\beta\le b\le1$ and $0<x<1$, two prices give revenue at least
\begin{equation}\label{eq:tf_menu}
 b\dot A(x)-(1-b)\ell+\min\{b(c-A(x)),T_b\}
\end{equation}
in the supremum sense.

Price products in $C_b$ at $b$ and all other products at $Q=b/x>b$. Start the expensive-product signals with the states in $\mathcal H\cap\{bY\ge Q\}$, grouping them by their favorite $f$. These favorites are expensive: a cheap favorite would satisfy
\[
 bY\le b\eta-\Delta=w-(1-b)\eta\le b<Q.
\]
The total surplus $bY-Q$ available for pooling is $QA(x)$. Draw a fractional submeasure $B$ from $\mathcal L$ with total cost
\[
 \int_B[Q+(1-b)\eta]\,d\mu=QA(x),
\]
using $\mathcal L\setminus L_b$ first. Such a choice exists, including at atoms, because
\[
 QA(x)\le b\mathbb E[(\eta-1)\mathbf1_{\mathcal H}]
 =b(c-\ell)<Qc.
\]
Distribute each selected low state across the expensive-product groups in proportion to their surpluses. Since $v_f-(1-b)\eta\ge bY$ on each initial high state and $v_f\ge0$ on each selected low state, each resulting group satisfies
\[
 \overline v_f\ge Q+(1-b)\overline\eta.
\]
Combining this inequality with the envelope for $f$ gives
$Q\le a_f+(b-a_f)\overline\eta$, so $\overline\eta>1$. For every cheap product $j$,
\[
 \overline v_j-b\le(1-b)\overline\eta\le\overline v_f-Q.
\]
Thus the buyer prefers an expensive product to every cheap product and to no purchase.

Let $t=\mu(B)$ and $\ell_B=\int_B\eta\,d\mu$. The expensive signals have initial mass $x\dot A(x)-A(x)$ and earn
\[
 Q\bigl(x\dot A(x)-A(x)+t\bigr)
 =b\dot A(x)-(1-b)\ell_B
 \ge b\dot A(x)-(1-b)\ell.
\]
Also $t\le A(x)$. On the remaining low states, group $L_b\setminus B$ by their cheap favorites and use their surplus above $b$ to pool as many other remaining states as possible. Nonnegative values make the resulting revenue at least
\[
 \min\left\{b(c-t),\int_{L_b\setminus B}w\,d\mu\right\}
 =\min\{b(c-t),T_b\}
 \ge\min\{b(c-A(x)),T_b\}.
\]
The equality follows from the priority rule: either no state in $L_b$ was used in $B$, or all remaining low states belong to $L_b$ and are worth at least $b$. Pool unused states in one residual signal. This proves \eqref{eq:tf_menu}. A common proportional discount makes every credited purchase strictly preferable to no purchase and makes the expensive signals strictly prefer the expensive price level. Taking the discount to zero covers any tie convention.

\paragraph{3. The revenue ceiling bounds the value profile.}
Write $z=(w-\eta)/(1-\eta)$ on $\{\eta<1\}$. Integrating the inverse profile gives
\[
 J_b:=\int_0^c\log\frac{b(y)}\beta\,dy
 =\int_\beta^1\frac{T_t}{t^2}\,dt.
\]
Indeed, for $t>\beta$ the set where $b(y)>t$ has length $T_{t+}/t<c$, and endpoints of tails do not affect the integral. On $\mathcal L$,
\[
 w-\eta\le\beta(1-\eta)
       +w\mathbf1_{\{\eta<1,z>\beta\}}(1-\beta/z).
\]
Below the cut this follows from $z\le\beta$; above it, the slack is $\eta(1-\beta/z)\ge0$. Integrating gives
\begin{equation}\label{eq:tf_low_bounds}
 V\le\beta(c-\ell+J_b)-D.
\end{equation}

Set
\[
 d=A(1),\qquad \delta=c-d,\qquad K=1-\delta.
\]
Comparing $(Y-1)_+$ with $\eta-1$ on $\mathcal H$ gives
\begin{equation}\label{eq:tf_delta}
 c-\ell-D/\beta\le d\le c-\ell,
 \qquad \ell\le\delta\le\ell+D/\beta.
\end{equation}
Moreover, $0\le\dot A(x)\le d+\Pr(\mathcal H)=K$, and $\dot A(x)=0$ for sufficiently small $x$ because the reduced prior is finite. From \eqref{eq:tf_menu} and $\delta\ge\ell$, the revenue ceiling requires
\[
 R+\delta\ge b(\dot A(x)+\delta)+\min\{b(c-A(x)),T_b\}
 \quad(\beta\le b\le1).
\]
Define the upper bound on this derivative by
\begin{equation}\label{eq:tf_G}
 G_0(y)=\min\left\{1,\inf_{\beta\le b\le1}
 \left[\frac{R+\delta}{b}-\min\{y,T_b/b\}\right]\right\},
 \qquad 0\le y\le c.
\end{equation}
It is positive, nonincreasing, and $1$-Lipschitz, with
\[
 G_0(y)\ge\min\{1,R+\delta-y\},\qquad
 G_0(y)=1\ (y\le\delta),\qquad
 \dot A(x)+\delta\le G_0(c-A(x)).
\]
On $[\delta,c]$, $G_0-\delta\ge\min\{K,R-c\}>0$. Integrating the last differential inequality along $A$, which starts at zero and ends at $d$, yields
\begin{equation}\label{eq:tf_time}
 \int_\delta^c\frac{dy}{G_0(y)-\delta}<1.
\end{equation}
Strictness follows from $\dot A=0$ on an initial interval; if $d=0$, the integral is zero. Feasibility of every inverse cut in \eqref{eq:tf_G}, together with \eqref{eq:tf_low_bounds}--\eqref{eq:tf_delta}, gives the two bounds we need:
\begin{equation}\label{eq:tf_profile_bounds}
 b(y)\le\min\left\{1,\frac{R+\delta}{y+G_0(y)}\right\},\qquad
 V\le\beta\left(d+\int_0^c\log\frac{b(y)}\beta\,dy\right).
\end{equation}

\paragraph{4. Reducing the bound to four simple pieces.}
We now maximize the welfare bound in \eqref{eq:tf_profile_bounds} subject to the restriction \eqref{eq:tf_time}. First, $b(y)\le1$ implies
\begin{equation}\label{eq:tf_cap}
 V\le\max_{0\le\beta\le1}\beta\bigl[d+c\log(1/\beta)\bigr]
 =ce^{-\delta/c}\le c-\frac\delta2=\frac{c+d}{2}.
\end{equation}
The last inequality uses $e^{-t}\le1-t/2$ for $0\le t\le1$. If $d\le R-1$, this would give $V<R/2$, contradicting \eqref{eq:tf_contradiction}. Thus $d>R-1$.

Shift the interval by writing $u=y-\delta$, and set
\[
 H_0(u)=G_0(\delta+u)-\delta,\qquad
 H_{\min}(u)=\min\{K,R-\delta-u\},\qquad 0\le u\le d.
\]
Here $H_{\min}\le H_0\le K$ and $\int_0^dH_0^{-1}<1$. By contrast, with $s=(R-1)/K$, the reciprocal integral of $H_{\min}$ is
\[
 \tau_{\min}=s+\log\frac K{R-c}>1.
\]
Indeed, $\tau_{\min}\le1$ would imply $c\le R-Ke^{s-1}$ and, by \eqref{eq:tf_cap},
\[
 V\le R-\frac\delta2-Ke^{s-1}
 \le R(1-e^{-1})+\delta(e^{-1}-1/2)<\frac23R,
\]
where we used $e^s\ge1+s$ and $2<e<3$.

A convex combination $H$ of $H_{\min}$ and $H_0$ therefore has reciprocal integral one. It is positive, nonincreasing, and $1$-Lipschitz. Replacing $H_0$ by $H$ in
\[
 b(\delta+u)\le\min\left\{1,\frac{R+\delta}{u+2\delta+H_0(u)}\right\}
\]
strictly raises the log integral of this upper profile: wherever $H_0>H_{\min}$, the bound is below one and increases strictly.

We can now replace $H$ by a function with only four pieces. Put $r=H(d)$ and $L=d+r$. The Lipschitz bound implies $L>R-\delta$; otherwise $H=H_{\min}$, whose reciprocal integral exceeds one. Choose $k\in[r,K]$ so that
\[
 H_*(u)=\max\{\min(K,R-\delta-u),\min(k,L-u)\}
\]
has reciprocal integral one. Existence follows from continuity: at $k=r$ this function lies below $H$, and at $k=K$ it lies above $H$. Its pieces are a plateau at $K$, a ramp of slope $-1$, a plateau at $k$, and a final ramp ending at $r$. The breakpoints are $R-1$, $R-\delta-k$, and $L-k$.

This replacement weakly increases the log integral of the upper profile. To see why, apply the convex tangent inequality to
\[
 f_u(z)=\max\left\{0,\log\frac{u+2\delta+1/z}{R+\delta}\right\},
 \qquad z=1/H(u),\quad z_*=1/H_*(u).
\]
The difference $z-z_*$ is zero on the initial plateau, nonpositive on the first ramp, changes sign at most once on the middle plateau, and is nonnegative on the final ramp. Its integral is zero. On the last three pieces, a nondecreasing tangent slope is $-H_*^2/(u+2\delta+H_*)$. Multiplying this slope by $z-z_*$ and integrating therefore gives a nonnegative value. Convexity implies $\int f_u(1/H)\,du\ge\int f_u(1/H_*)\,du$, as claimed.

Let $p=(L-(R-\delta))/k>0$ be the length of the middle plateau divided by its height. Integrating $1/H_*$ gives
\[
 1=s+p+\log(K/r),\qquad
 0<p\le1-s,\qquad
 r=K e^{s+p-1},\qquad c=R+kp-r.
\]
The resulting upper profile is one up to $y=R-k$, equals $(R+\delta)/(y+\delta+k)$ up to $y=R+kp-k$, and is constant thereafter. Its endpoint $B$ and log integral relative to that endpoint are
\begin{equation}\label{eq:tf_BI}
 B=\frac{R+\delta}{R+\delta+kp},\qquad
 I=kp-(\delta+k)\log\left(1+\frac{kp}{R+\delta}\right).
\end{equation}
Since $\beta\le B$, \eqref{eq:tf_profile_bounds} now yields
\begin{equation}\label{eq:tf_U}
 V<\beta\left(d+I+c\log\frac B\beta\right)
 \le\mathcal U(R):=\max_{0\le\rho\le1}
 B\rho\bigl[d+I+c\log(1/\rho)\bigr].
\end{equation}
We use the continuous extension at $\rho=0$.

At $R=1$, \cref{lem:tf_scalar} bounds this maximum by $2/3$. To compare other values of $R$ with this benchmark, hold $\delta,k,p$ fixed and let $q$ range from $1$ to $R$, setting
\[
 n(q)=e^{(q-1)/K+p-1},\qquad
 c(q)=q+kp-Kn(q),\qquad d(q)=c(q)-\delta.
\]
Define $B(q),I(q),\mathcal U(q)$ by \eqref{eq:tf_BI}--\eqref{eq:tf_U}, with $q$ in place of $R$. Along this path, $1/3<n(q)\le1$, $0<d(q)\le c(q)<1$, $0\le I(q)\le kp\le1$, and $\delta+k\le1$. Thus
\[
 c'=1-n<\frac23,\qquad
 B'=\frac{kp}{(q+\delta+kp)^2}\le\frac14,\qquad
 BI'=\frac{(\delta+k)kp}{(q+\delta+kp)^2}\le\frac14.
\]
For each fixed $\rho$, the expression maximized in \eqref{eq:tf_U} has derivative at most
$2B'+Bc'+BI'<3/2$: use $d=c-\delta$, $c,I\le1$, and $\rho(1+\log(1/\rho))\le1$. The maximum therefore increases by at most $3(R-1)/2$. Applying \cref{lem:tf_scalar} with $\alpha=k/K$ gives
\[
 V<\mathcal U(R)\le\frac23+\frac32(R-1)
 \le\frac53R-1=V,
\]
the desired contradiction.

\paragraph{5. Implementing coarse horizontal disclosure.}
We have shown that uniform pricing or one of the constructed menus exceeds $3W/5$. It remains to check that its signals can be labeled by distinct sets of original favorite products.

For a uniform price, any affordable posterior can be decomposed into affordable posteriors supported on at most two reduced states: these are the vertices of the simplex intersected with a single affordability inequality. A two-state posterior with both states individually affordable can instead be split into singletons. Otherwise it combines a state with $w\ge p$ and one with $w<p$. For a fixed pair, the mixtures at which some product is affordable form an interval containing the high-value endpoint. Consequently, all signals supported on that pair can be merged without losing affordability. This gives distinct singleton or pair labels for a full-sale uniform price.

For the two-price construction, each expensive signal contains its own initial high state, which appears in no other signal. Each cheap signal contains its own initial state in $L_b\setminus B$ and, apart from that state, only low states outside $L_b$. Hence cheap-signal supports are distinct from one another and from the expensive-signal supports. The residual signal contains none of these initial states, so its support is distinct as well. Label each signal by its support. The resulting kernel depends only on $i^*(v)$ and reports a set containing that product, exactly as required by \cref{def:coarse_horizontal}. A sufficiently small common discount preserves the strict revenue improvement under arbitrary ties.

The initial pooling applies to any integrable prior and preserves $W$ exactly, so no limiting argument is needed. If $W=0$, the claim is immediate. Finally, $\optrev\le\optwel$ gives the revenue comparison in the theorem.
\end{proof}

\begin{lemma}[Scalar bound]\label{lem:tf_scalar}
For any $(\delta,p,\alpha)\in[0,1]^3$, define
\begin{gather*}
 K=1-\delta,\quad k=K\alpha,\quad a=kp,\quad
 c=1+a-Ke^{p-1},\quad d=c-\delta,\\
 B=\frac{1+\delta}{1+\delta+a},\qquad
 I=a-(\delta+k)\log\left(1+\frac a{1+\delta}\right).
\end{gather*}
Then
\[
 B\max_{0\le\rho\le1}\rho\bigl[d+I+c\log(1/\rho)\bigr]\le\frac23,
\]
where the expression at $\rho=0$ is defined by continuity.
\end{lemma}

\begin{proof}[Verification of \cref{lem:tf_scalar}]
We reduce the claim to the nonnegativity of two polynomials.

Put $A=1+\delta$ and $T=2A+a$. The elementary bounds
\[
 \exp(p-1)\ge N:=\frac{2+3p+p^3}{6},\qquad
 \log\left(1+\frac aA\right)\ge L_0:=\frac{2a}{T}+\frac{2a^3}{3T^3}
\]
follow, respectively, from the cubic Taylor bound for the exponential and the first two terms of the positive series
$\log((1+z)/(1-z))=2\sum_{j\ge0}z^{2j+1}/(2j+1)$, with $z=a/T$.
Define
\[
 c_0=1+a-KN,\qquad d_0=c_0-\delta,\qquad
 I_0=a-(\delta+k)L_0.
\]
Then $c_0\ge c$, $d_0\ge d$, and $I_0\ge I\ge0$, so replacing $c,d,I$ by these bounds can only increase the claimed maximum. If $I_0\ge\delta$, the maximizing $\rho$ is one. Otherwise, the maximum before multiplication by $B$ is $c_0e^{-t}$ for $t=(\delta-I_0)/c_0\in[0,1]$, since $c_0\ge\delta$ and $I_0\ge0$. The bound $e^{-t}\le1-t/2$ therefore gives
\begin{gather*}
 B\max_{0\le\rho\le1}\rho\bigl[d+I+c\log(1/\rho)\bigr]
 \le\max\{U_0,U_1\},\\
 U_0=B(d_0+I_0),\qquad U_1=\frac B2(c_0+d_0+I_0).
\end{gather*}
It is therefore enough to show $U_0,U_1\le2/3$.

Multiplying $2/3-U_j$ by the positive common denominator $6(A+a)T^3$ gives the following polynomials:
\begin{align*}
 \Pi_0={}&\bigl[4(A+a)-AK(4-3p-p^3)-12Aa\bigr]T^3\\
         &\quad+4A(\delta+k)a(3T^2+a^2),\\
 \Pi_1={}&\bigl[4(A+a)-3A(2-\delta+3a)+AK(2+3p+p^3)\bigr]T^3\\
         &\quad+2A(\delta+k)a(3T^2+a^2).
\end{align*}
We verify that both are nonnegative on $(\delta,p,\alpha)\in[0,1]^3$, with $K=1-\delta$, $k=K\alpha$, and $a=kp$. Their coordinatewise degrees are at most $(5,6,4)$.

The check uses the Bernstein basis
$\mathcal B_r^n(t)=\binom nr t^r(1-t)^{n-r}$.
These functions are nonnegative and sum to one. On a rectangle with endpoints $l_j,u_j$, substitute $(\delta,p,\alpha)_j=l_j+(u_j-l_j)t_j$ and expand the transformed polynomial as $\sum_s c_s\prod_{j=1}^3t_j^{s_j}$. Its Bernstein coefficients of degree $n=(5,6,4)$ are
\begin{equation}\label{eq:tf_bernstein}
 b_r=\sum_{s\le r}c_s\prod_{j=1}^3
       \frac{\binom{r_j}{s_j}}{\binom{n_j}{s_j}},
 \qquad 0\le r_j\le n_j.
\end{equation}
Thus the polynomial is a convex combination of the $b_r$, and is nonnegative whenever these coefficients are nonnegative.

Substitution in \eqref{eq:tf_bernstein} gives the following rational lower bounds on every coefficient. The first eight rectangles cover the cube for $\Pi_0$; the last three cover it for $\Pi_1$.
\begin{center}
\renewcommand{\arraystretch}{1.15}
\begin{tabular}{ccccc}
\hline
Polynomial & $\delta$ & $p$ & $\alpha$ & Lower bound\\
\hline
$\Pi_0$ & $[0,1]$ & $[0,1]$ & $[0,3/4]$ & $0$\\
$\Pi_0$ & $[0,1]$ & $[0,1/2]$ & $[3/4,1]$ & $0$\\
$\Pi_0$ & $[1/6,1]$ & $[1/2,1]$ & $[3/4,1]$ & $1/50$\\
$\Pi_0$ & $[0,1/12]$ & $[1/2,3/4]$ & $[3/4,1]$ & $1/50$\\
$\Pi_0$ & $[1/12,1/6]$ & $[1/2,3/4]$ & $[3/4,1]$ & $1/50$\\
$\Pi_0$ & $[0,1/24]$ & $[3/4,1]$ & $[3/4,1]$ & $1/50$\\
$\Pi_0$ & $[1/24,1/12]$ & $[3/4,1]$ & $[3/4,1]$ & $1/50$\\
$\Pi_0$ & $[1/12,1/6]$ & $[3/4,1]$ & $[3/4,1]$ & $1/50$\\
\hline
$\Pi_1$ & $[0,1/3]$ & $[0,1/2]$ & $[0,1]$ & $0$\\
$\Pi_1$ & $[0,1/3]$ & $[1/2,1]$ & $[0,1]$ & $1/50$\\
$\Pi_1$ & $[1/3,1]$ & $[0,1]$ & $[0,1]$ & $1/50$\\
\hline
\end{tabular}
\end{center}
All $11\cdot6\cdot7\cdot5=2{,}310$ coefficients are nonnegative. Hence $U_0,U_1\le2/3$, proving the lemma.
\end{proof}

\subsection{Partially Informed Buyer}
\label{subapx:partially_informed_buyer}

We first provide a simple implication of the likelihood-ratio bound.

\begin{lemma}
\label{lem:partially_informed_buyer_posterior}
Fix an initial experiment $\eta:\vals\to\Delta(\mathcal{T})$ and a seller information structure $(\signals,\experiment)$. 
Suppose $t$ satisfies $\prob[\val\sim\dist,\;t'\sim\eta(\val)]{t'=t}>0$ and $L_{\dist,\eta}(t)\leq e^\lambda$. 
Then for every seller signal $\signal$ and every nonnegative function $g:\vals\to\reals_+$,
\begin{align*}
\expect[\dist,\eta,\experiment]{g(\val)\given t,\signal}
\geq
e^{-\lambda}
\expect[\dist,\experiment]{g(\val)\given \signal},
\end{align*}
whenever the conditional expectations are well defined.
\end{lemma}

\begin{proof}
Since $L_{\dist,\eta}(t)\leq e^\lambda$ and $t$ has positive ex-ante probability, we have $\eta(t\mid\val)>0$ for every $\val\in\vals$. 
Let
\begin{align*}
a_t=\min_{\val\in\vals}\eta(t\mid\val)
\quad\text{and}\quad
b_t=\max_{\val\in\vals}\eta(t\mid\val).
\end{align*}
Then $a_t>0$ and $b_t/a_t\leq e^\lambda$.

Using the conditional independence of $t$ and $\signal$ given $\val$, we have
\begin{align*}
\expect[\dist,\eta,\experiment]{g(\val)\given t,\signal}
=
\frac{
\expect[\dist,\experiment]{g(\val)\eta(t\mid\val)\given \signal}
}{
\expect[\dist,\experiment]{\eta(t\mid\val)\given \signal}
}.
\end{align*}
Since $g$ is nonnegative,
\begin{align*}
\expect[\dist,\experiment]{g(\val)\eta(t\mid\val)\given \signal}
\geq
a_t\expect[\dist,\experiment]{g(\val)\given \signal},
\end{align*}
and $\expect[\dist,\experiment]{\eta(t\mid\val)\given \signal}\leq b_t$.
Therefore,
\begin{equation*}
\expect[\dist,\eta,\experiment]{g(\val)\given t,\signal}
\geq
\frac{a_t}{b_t}
\expect[\dist,\experiment]{g(\val)\given \signal}
\geq
e^{-\lambda}
\expect[\dist,\experiment]{g(\val)\given \signal}.\qedhere
\end{equation*}
\end{proof}

\begin{proof}[Proof of \cref{prop:partially_informed_buyer}]
Fix any $\epsilon>0$. 
By the proof of \cref{prop:uniform_pricing}, there exists a coarse horizontal disclosure $(\signals,\experiment)$ and a uniform price $\pay$ such that, in the baseline model, the product is sold with probability one and
\begin{align*}
\pay\geq \frac{1}{2}\optwel-\epsilon.
\end{align*}
Thus, for every seller signal $\signal$ in the support of $\experiment$,
\begin{align*}
\max_{i\in[m]}
\expect[\dist,\experiment]{\val_i\given \signal}
\geq
\pay.
\end{align*}

Consider the same coarse horizontal disclosure in the partially informed buyer environment, and post the uniform price
\begin{align*}
\hat{\pay}=e^{-\lambda}\pay.
\end{align*}
For every realized initial signal $t$ with positive ex-ante probability and $L_{\dist,\eta}(t)\leq e^\lambda$, applying \cref{lem:partially_informed_buyer_posterior} to $g(\val)=\val_i$ gives
\begin{align*}
\belief_i^\eta(t,\signal)
=
\expect[\dist,\eta,\experiment]{\val_i\given t,\signal}
\geq
e^{-\lambda}
\expect[\dist,\experiment]{\val_i\given \signal}.
\end{align*}
Hence,
\begin{align*}
\max_{i\in[m]}\belief_i^\eta(t,\signal)
&\geq
e^{-\lambda}
\max_{i\in[m]}
\expect[\dist,\experiment]{\val_i\given \signal}\\
&\geq
e^{-\lambda}\pay
=
\hat{\pay}.
\end{align*}
Therefore, the buyer purchases whenever $L_{\dist,\eta}(t)\leq e^\lambda$ for every realized signal $t$ with positive ex-ante probability.
Since $\eta$ is $(\lambda,\beta)$-approximately bounded relative to $\dist$,
\begin{align*}
\prob[\val\sim\dist,\;t\sim\eta(\val)]
{
L_{\dist,\eta}(t)\leq e^\lambda
}
=
1-\beta_{\dist,\eta}(\lambda)
\geq
1-\beta.
\end{align*}
The expected revenue of the constructed mechanism is therefore at least
\begin{align*}
\rev_\eta(\mech)
&\geq
\hat{\pay}(1-\beta)\\
&=
e^{-\lambda}\pay(1-\beta)\\
&\geq
e^{-\lambda}(1-\beta)\left(\frac{1}{2}\optwel-\epsilon\right).
\end{align*}
Since $\epsilon>0$ is arbitrary, the result follows.
\end{proof}

\subsection{Optimal Pricing with Two Products}
\label{subapx:two_item_optimal}
\begin{proof}[Proof of \cref{thm:two_item_optimal}]

We first consider finite priors and finite signal sets. The main step is to pool all signals except the one buying the more expensive product. We show that this pooling improves revenue whenever partial allocations remain. We then extend the result to integrable priors.

\paragraph{1. Selecting an optimum.}
Fix a finite prior and any finite bound on the number of signals. An optimum exists within this class. Removing negative-payment options cannot lower revenue, and IR bounds chosen prices by the largest value $M$. Contracts and joint state--signal probabilities therefore lie in a compact set. Multiplying IC and IR by signal probabilities makes these constraints closed, even at zero-probability signals. Choose a revenue optimum with the fewest signals and, among these, the highest probability of sale. Suppose its revenue exceeds that of every pricing mechanism. All transformations below use no more signals.

Write $v_i,q_i,p_i,U_i$ for a signal's posterior mean, allocation, payment, and utility, and $\mu_i$ for its probability. Every payment is positive. Indeed, a zero-payment signal of probability $\alpha<1$ and mean $v_0$ can be removed by replacing each other signal by
\[
 \mu'_i=\frac{\mu_i}{1-\alpha},\qquad
 v'_i=(1-\alpha)v_i+\alpha v_0,\qquad
 p'_i=(1-\alpha)p_i+\alpha v_0\cdot q_i.
\]
Rents and utility differences are multiplied by $1-\alpha$, and revenue weakly increases, contradicting the minimal signal count. An optimum with only a zero-payment signal is already matched by no-information pricing. Also, every binding IC constraint between distinct signals points toward a strictly lower price: switching to a higher price raises revenue, while switching to an equal price allows the two signals to be pooled.

\paragraph{2. Reducing the options that always allocate a product.}

Replace the highest-price options by the two products at that same price. Their buyers can obtain at least as much utility; any other buyer who switches pays weakly more. No old signal needs to be split. There is now an option that allocates a product with probability one. Let $c_0$ be the lowest price of such an option, with allocation $(z,1-z)$, and let $u_0(v)$ be the maximum utility from all cheaper options and outside.

Partition the signals paying at least $c_0$ into $H=\{v_1\geq v_2\}$ and $L=\{v_2>v_1\}$. For a signal in $H$, its old payment satisfies
\[
 p\leq v_1-u_0(v),\qquad p\leq c_0+(1-z)(v_1-v_2);
\]
the analogous bounds on $L$ are $p\leq v_2-u_0(v)$ and $p\leq c_0-z(v_1-v_2)$. Pool each nonempty group. With $d_H=\mathbb E[v_1-v_2\mid H]\geq0$ and $d_L=\mathbb E[v_1-v_2\mid L]<0$, offer the respective products at
\[
 P_1=\min\{\mathbb E[v_1-u_0(v)\mid H],c_0+(1-z)d_H\},\qquad
 P_2=\min\{\mathbb E[v_2-u_0(v)\mid L],c_0-zd_L\}.
\]
Each price is at least $c_0$ and at least its group's average old payment. Convexity of $u_0$ ensures IC against all cheaper options after pooling. The bounds $d_L\leq P_1-P_2\leq d_H$ ensure IC between the two products. If a group is empty, omit its product. Cheaper buyers who switch to these options pay more. Pooling identical choices after any switching leaves at most two pure allocations, each more expensive than every remaining allocation. The transformation weakly raises revenue and sale probability without increasing the number of signals, so it preserves our choice of optimum. Call the remaining allocations \emph{partial}: they satisfy $0<q_{i1}+q_{i2}<1$.

\paragraph{3. The binding comparison between purchase options.}

For the fixed posterior distribution, linear-programming duality gives nonnegative IC multipliers $f_{ij}$ and IR multipliers $\lambda_i$ such that
\begin{equation}\label{eq:two_dual}
 o_i:=\mu_i+\sum_j f_{ji}=\lambda_i+\sum_j f_{ij},\qquad
 \Phi_i=o_iv_i-\sum_jf_{ji}v_j.
\end{equation}
The allocation $q_i$ maximizes $\Phi_i\cdot q$ over the allocation simplex, revenue equals $\sum_i\Phi_i\cdot q_i$, and positive multipliers require binding constraints. These conditions follow by substituting $p_i=v_i\cdot q_i-U_i$ into the objective and attaching multipliers to
$U_i\geq0$ and $U_i-U_j\geq(v_i-v_j)\cdot q_j$. In particular, a partial allocation contributes zero virtual surplus: its capacity is slack, so both coordinates of $\Phi_i$ are nonpositive and every used coordinate is zero.

If there were only one pure allocation, say product $j$ at posterior $A$ of probability $r$, it would have no incoming binding constraint. Revenue would equal $rA_j\leq\mathbb E[v_j]$, which no-information pricing can achieve. Thus there are two pure allocations, and there is at least one partial allocation; otherwise the mechanism is already pricing. Denote the pure signals by $A,B$, with probabilities $r,s>0$.

Scale all partial allocations and prices, and both pure rents, by a common factor near one. Because comparisons toward higher prices are strict, this preserves IC and IR locally if the two pure cross comparisons are strict, or if a binding comparison has equal rents. Revenue is affine in the factor, so optimality makes it constant; increasing the factor then raises sale probability. Hence one pure cross comparison binds and the rents differ. Relabel products so that $A$ buys product 1 at the higher price and $B$ buys product 2. Write
\[
 A=(A_1,a),\quad B=(B_1,b),\quad P=p_B,\qquad
 p_A=A_1-a+P,\quad U_A=a-P,\quad U_B=b-P.
\]
Only $A$'s comparison with $B$ can bind, and $a\ne b$. If $a<b$, pure IC implies that both coordinates of $B-A$ are at most $b-a$. Every partial allocation $q$ therefore satisfies $(B-A)\cdot q<b-a=U_B-U_A$. Thus $B$ strictly prefers its contract to every partial option, as well as to $A$ and outside, so its price could increase. Consequently,
\[
 A_1>a>b\geq P>0.
\]

Let $\gamma=f_{AB}\geq0$. No other binding constraint points toward a pure option, so \eqref{eq:two_dual} gives $\Phi_A=rA$ and $\Phi_B=(s+\gamma)B-\gamma A$. Let $T$ be total revenue from partial signals. Separate the product-1 price premium from revenue by writing
\begin{equation}\label{eq:two_revenue}
 N:=(r+s)P+T=ra+sb-\gamma(a-b),\qquad
 R=r(A_1-a)+N.
\end{equation}
Every partial price is below $P$, so $N<P$. Since $R>\mathbb E[v_1]\geq rA_1$, we also have $N>ra$.

\paragraph{4. A bound on allocated value.}

Fix a partial signal $C=(x,y)$ with contract $(q,p)$, rent $U$, and allocated value $\omega=C\cdot q=p+U$. Its strict comparison with the higher-priced product 2 gives $z=P-y+U>0$. Move probability mass $\varepsilon$ of this signal into $A$, and multiply every partial allocation and price by
\[
 \tau=\frac{r(a-b)+\varepsilon(y-b)}{r(a-b)+\varepsilon(U-U_B)}
      =1-\frac{\varepsilon z}{r(a-b)+\varepsilon(U-U_B)}.
\]
Set $A'=(rA+\varepsilon C)/(r+\varepsilon)$, $p'_B=b-\tau U_B$, and $p'_A=A'_1-A'_2+p'_B$. The pure comparison at $A'$ stays tied. Its utility equals
$\tau(rU_A+\varepsilon U)/(r+\varepsilon)$, which bounds every partial deviation by averaging old IC. The bounds at $B$, partial IC, and IR all scale by $\tau$; strict comparisons toward higher prices persist for small $\varepsilon>0$.

Using $T-(r+s)U_B=(r-\gamma)(a-b)$, the revenue derivative at zero is
\[
 x-y+P-p-\frac{z}{r(a-b)}\bigl[T-(r+s)U_B\bigr]
 =x-\omega+\frac\gamma r z.
\]
Optimality therefore gives $\omega\geq x$. Since $p>0$ and $q_1+q_2<1$, this also implies $q_2>0$.

\paragraph{5. Bounding partial revenue and pooling.}

Summing the flow balances in \eqref{eq:two_dual}, and using $\Phi_{i2}=0$ at every partial signal, gives
\[
 \sum_i\lambda_i=1,\qquad
 \sum_i\lambda_i v_{i2}=\sum_i\Phi_{i2}=N.
\]
Only zero-rent signals have positive $\lambda_i$. IC at those signals gives $p\geq v_i\cdot q$ for every offered contract. Averaging these inequalities yields
\[
 p\geq\sum_i\lambda_i v_i\cdot q\geq Nq_2.
\]

For a partial contract, put $h=1-q_1-q_2>0$. Using $\omega\geq x$ and $\omega\geq p$,
\[
 y-\omega=\frac{h\omega+q_1(\omega-x)}{q_2}
 \geq\frac{hp}{q_2}\geq hN.
\]
IC at $A$ gives $P-p\leq ah-(A_1-a)q_1<ah$ when $q_1>0$. If $q_1=0$, IC at $B$ instead gives $P-p\leq bh<ah$. Thus
\begin{equation}\label{eq:two_gap}
 y>p+\frac Na(P-p).
\end{equation}

Pool all signals except $A$, and denote their mean by $(\bar x,\bar b)$. Write $m=1-r-s>0$ and $Y=\sum_{i\text{ partial}}\mu_i v_{i2}$. Summing \eqref{eq:two_gap} and using \eqref{eq:two_revenue},
\begin{align*}
 (1-r)\bar b=s b+Y
 &>s b+T+\frac Na(Pm-T)\\
 &=(1-r)N+s(b-P)+(P-N)\left(\frac Na-r\right)\\
 &\geq(1-r)N,
\end{align*}
where the last inequality uses $N\leq P\leq b$ and $N>ra$. Meanwhile, $R>\mathbb E[v_1]=rA_1+(1-r)\bar x$ implies
$\bar x<(N-ra)/(1-r)<N$. Post prices
\[
 Q_2=\min\{a,\bar b\}>N,\qquad Q_1=A_1-a+Q_2.
\]
At $A$, product 1 is a best response with nonnegative utility. At the pooled signal, product 2 gives nonnegative utility and product 1 gives negative utility. Revenue is
\[
 rQ_1+(1-r)Q_2=r(A_1-a)+Q_2>R,
\]
contradicting optimality. This proves the claim for every finite prior and finite signal set.

\paragraph{6. Integrable priors and attainment.}
Finite menus approximate every feasible revenue under an integrable prior. To see this, fix a design, let $b$ be its posterior mean, $(q,p)$ its selected contract, and $U=b\cdot q-p$ its utility. On
$G_K=\{\max_j b_j\leq K,\ p\geq0\}$, selected prices lie in $[0,K]$. Choose a finite $\varepsilon$-net of these menu contracts, including outside, in the metric $K\|q-q'\|_1+|p-p'|$, and discount all prices by $1-\delta$. If the newly selected contract had original price $\widehat p$, comparison with the approximating old contract gives
\[
 U+\delta\widehat p\geq U+\delta p-\varepsilon,
 \qquad \widehat p\geq p-\varepsilon/\delta.
\]
All new payments are nonnegative, so new revenue is at least
$(1-\delta)(\mathbb E[p\mathbf1_{G_K}]-\varepsilon/\delta)$.
Since $p_+\leq\max_j b_j$ is integrable, this bound approaches $\mathbb E[p_+]\geq\mathbb E[p]$ as $K\to\infty$, $\delta\to0$, and $\varepsilon/\delta\to0$. Pooling signals that choose the same contract produces a finite prior of posterior means. Apply the finite-prior result and compose the resulting experiment with the original one. This preserves means and revenue, so both revenue suprema agree.

For attainment, if the revenue supremum is zero, one no-information signal and a free product suffice. Otherwise, pool each pricing design by the chosen product or outside, then absorb the outside signal as in Step 1. Pure allocations stay pure, so at most two purchase signals suffice. Take a maximizing sequence and represent its two recommendations by measures $\nu_1^n,\nu_2^n\geq0$ summing to the prior $F$. Domination by $F$ gives a weakly convergent subsequence whose limits still sum to $F$. Masses and first moments converge: for the latter, integrability gives the uniform tail bound
\[
 \int(v_j-K)_+\,d\nu_i^n\leq\int(v_j-K)_+\,dF\longrightarrow0.
\]
Pass also to a subsequence along which each price converges in $[0,\infty]$. A diverging price forces its recommendation's probability to vanish by IR. Moment convergence then gives
\[
 0\leq P_i^n\nu_i^n(\mathbb R_+^2)\leq\int v_i\,d\nu_i^n\longrightarrow0.
\]
Omit any such product. All remaining IC and IR inequalities pass to the limit because they involve only converging prices, masses, and first moments. An omitted product can be assigned a finite price above its value at each remaining posterior mean. The limit therefore attains optimal revenue with finite prices and at most two purchase signals.

\end{proof}

\subsection{Welfare Bound with Two Products}
\label{subapx:two_item_approx}
We first establish the upper bound on revenue relative to welfare, using \cref{thm:two_item_optimal} to restrict attention to pricing. We then prove the matching guarantee under coarse horizontal disclosure.
\begin{lemma}
\label{lem:two_item_hardness}
For any \(\varepsilon > 0\), there exists an instance with two products where the optimal revenue is at most \(\frac23 + \varepsilon\) fraction of the optimal welfare.
\end{lemma}

\begin{proof}
For $0<\eta<1$, let values be $(1,0)$ with probability $1-\eta$ and $(1/\eta,2/\eta)$ with probability $\eta$. Welfare is $3-\eta$. We show that optimal revenue equals $2$, which no-information sale of product 2 achieves.

By \cref{thm:two_item_optimal}, it suffices to consider pricing with at most two purchase signals. Selling only one product earns at most its prior mean, hence at most $2$. Otherwise, let $\alpha$ be the probability of the signal recommending product 1, and let $t,u$ be the posterior probabilities of the high state at the two respective signals. IC implies $t\leq u$, so Bayes' rule gives
\[
 0\leq t\leq\eta\leq u\leq1,\qquad
 \alpha t+(1-\alpha)u=\eta,\qquad
 0\leq\alpha\leq\frac{1-\eta}{1-t}.
\]
IR at the first signal and IC at the second imply
\[
 p_1\leq1+(\eta^{-1}-1)t,\qquad
 p_2\leq p_1+(\eta^{-1}+1)u-1.
\]
Consequently, expected revenue is at most
\[
 1+\eta+(\eta^{-1}-1)t
 +\alpha\bigl[1-(\eta^{-1}+1)t\bigr].
\]
This expression is affine in $\alpha$. At the two endpoints of its allowed interval, it equals, respectively,
\[
 1+\eta+(\eta^{-1}-1)t\leq2,
 \qquad
 2-\frac{(1-\eta)t^2}{\eta(1-t)}\leq2.
\]
Thus optimal revenue is $2$. Its ratio to welfare tends to $2/3$ as $\eta\downarrow0$, proving the claim.
\end{proof}

\begin{proof}[Proof of \cref{prop:two_item_approx}]
To establish this result, we construct three distinct pricing mechanisms and show that the total revenue achieved by these mechanisms, in aggregate, is at least twice the optimal welfare.

First, consider the horizontal disclosure scheme. Let $\signal_1$ and $\signal_2$ denote the signals indicating that product 1 and product 2, respectively, have the highest value.
Recall that we denote $\belief(\signal)$ as the posterior mean value and $\sigprobsingle(\signal)$ as the ex ante probability of signal $\signal$.
Without loss of generality, we may assume that \(\belief_2(\signal_2) \geq \belief_1(\signal_1)\).
The optimal welfare can be expressed as:
\begin{align*}
    \optwel = \sigprobsingle(\signal_1) \belief_1(\signal_1) + \sigprobsingle(\signal_2) \belief_2(\signal_2).
\end{align*}
Now consider the following three different pricing mechanisms:
\begin{itemize}
    \item \(\mech_1\): Provide no information to the buyer and offer only the first product at a price equal to its expected value.
    \item \(\mech_2\): Provide no information to the buyer and offer only the second product at a price equal to its expected value.
    \item \(\mech_3\): Reveal to the buyer according to horizontal disclosure. Offer product prices $p_1 = \belief_1(\signal_1)$ and $p_2 = \belief_2(\signal_2) - \max\left\{0, \belief_1(\signal_2) - \belief_1(\signal_1)\right\}$.
\end{itemize}

We then analyze the revenue of each mechanism. For the first two mechanisms, the revenue corresponds to the expected value of the first or second product, respectively:
\[
\rev(\mech_1) = \sigprobsingle(\signal_1) \belief_1(\signal_1) + \sigprobsingle(\signal_2) \belief_1(\signal_2), \quad
\rev(\mech_2) \geq \sigprobsingle(\signal_2) \belief_2(\signal_2).
\]
Now consider the third mechanism. Note that given the construction of the prices, the buyer receiving signal $s_i$ will purchase product $i$ for $i\in\{1,2\}$, and the expected revenue is 
\begin{align*}
\rev(\mech_3) &= \sigprobsingle(\signal_1) \belief_1(\signal_1) + \sigprobsingle(\signal_2) \left(\belief_2(\signal_2) - \max\left\{0, \belief_1(\signal_2) - \belief_1(\signal_1)\right\} \right) \\
&\geq \sigprobsingle(\signal_1) \belief_1(\signal_1) + \sigprobsingle(\signal_2) \left(\belief_2(\signal_2) - \belief_1(\signal_2)\right).
\end{align*}

Combining the revenue bounds for all three mechanisms, we obtain:
\begin{align*}
    \rev(\mech_1) + \rev(\mech_2) + \rev(\mech_3) &\geq 2 \left( \sigprobsingle(\signal_1) \belief_1(\signal_1) + \sigprobsingle(\signal_2) \belief_2(\signal_2) \right) \\
    &= 2 \optwel.
\end{align*}
Therefore, at least one of the mechanisms among \(\mech_1\), \(\mech_2\), and \(\mech_3\) must achieve at least a~\(\frac{2}{3}\) fraction of the optimal welfare.
Finally, the tightness of the approximations is implied by \cref{lem:two_item_hardness}.
\end{proof}

\subsection{Full Surplus Extractions}
\label{apx:multi}

\label{subapx:full_surplus_multi}

We begin by providing a complete characterization in the case where the distribution is generic.
\begin{definition}[Generic Distribution]
A distribution $\dist$ is generic if for any $i\neq i'$, the total probability measure that $\val_i = \val_{i'}$ is zero. 
\end{definition}
\begin{proposition}
\label{prop:generic_full_surplus}
For any number of products and any \emph{generic} distribution \(\dist\), 
there exists a pricing mechanism with horizontal disclosure that extracts the full surplus if and only if for any \(i \neq i'\), it holds that
\begin{align*}
\expect[\dist]{ \val_i \given i^*(\val) = i} 
\geq \expect[\dist]{ \val_i \given i^*(\val) = i'}.
\end{align*}
\end{proposition}
\begin{proof}
We first prove the ``if" direction. 
Note that the if direction does not require the generic condition. 
Suppose the distribution \(\dist\) satisfies that 
\[
\expect[\dist]{ \val_i \given i^*(\val) = i} 
\geq \expect[\dist]{ \val_i \given i^*(\val) = i'}.
\]

In the following, we construct a horizontal disclosure together with a pricing mechanism that extracts the full surplus. Consider the horizontal disclosure \(\rbr{\signals, \experiment}\), where \(\signals = \{\signal_1,\signal_2,\cdots,\signal_m\}\) and \(\experiment\rbr{\val} = \signal_{i^{*}(\val)}\). Furthermore, consider the pricing mechanism in which the price of product \(i\), denoted by \(p_i\), is set as
\[
p_i = \expect[\dist]{\val_i \given i^*(\val) = i}.
\]

Notice that upon receiving signal \(\signal_i\), the buyer's belief about product \(i'\) is exactly
\[
\belief_{i'}\rbr{\signal_i} =  \expect[\dist]{ \val_{i'} \given i^*(\val) = i}.
\]
This implies that upon receiving signal \(\signal_i\), the buyer will always choose to purchase product~\(i\). Note that the price of product \(i\) is exactly the prior belief of product \(i\). This means that for any value profile \(\val\), the resulting allocation is efficient, and the buyer's utility under any signal \(\signal_i\) is always \(0\). Therefore, this information structure, together with the pricing mechanism, extracts the full surplus.

Next, we prove the ``only if'' direction. Suppose there exists an information structure together with a mechanism that extracts the full surplus. In order to extract the full surplus, the allocation must be efficient. Since the distribution is generic, this implies that for any value profile \(\val \in \supp\rbr{\dist}\), every signal \(\signal \in \supp\rbr{\experiment\rbr{\val}}\) that may be received given \(\val\) must induce the allocation of product \(i^{*}(\val)\). Due to the incentive compatibility constraint, the payments associated with these signals must be identical. Therefore, we can merge these signals without violating incentive compatibility, and the resulting revenue remains unchanged. This implies that after merging, there are at most \(m\) signals in total, where the \(i\)-th signal \(\signal_i\) corresponds to the allocation of product \(i\). For each value profile, the corresponding signal is simply the \(i^{*}(\val)\)-th signal. Hence, the information structure must be a horizontal disclosure. Moreover, observe that the mechanism must be a pricing mechanism. Since the mechanism extracts the full surplus, the buyer's utility under any signal \(\signal_i\) must be \(0\). This implies that the price $p_i$ for product $i$ must equal the prior belief of product $i$ upon receiving signal $i$:
\[
p_i = \belief_{i}\rbr{\signal_i}.
\]
Now consider the buyer's behavior. Since the buyer chooses to purchase product \(i\) upon receiving signal \(\signal_i\), it must be the case that purchasing any other product yields non-positive utility. Formally, for any \(i, i'\), it follows that
\[
\belief_{i'}\rbr{\signal_i} - p_{i'} \leq 0.
\]
Using the definition of \(p_i\) and the structure of horizontal disclosure, this implies
\[
\expect[\dist]{ \val_i \given i^*(\val) = i'} \leq \expect[\dist]{ \val_{i'} \given i^*(\val) = {i'}},
\]
which concludes the proof.
\end{proof}

\begin{proof}[Proof of \cref{lem:full_surplus}]
This is implied by the ``if'' direction in \cref{prop:generic_full_surplus} since in the proof of \cref{prop:generic_full_surplus}, the generic condition is not required for the ``if'' direction. 
\end{proof}

To prove \Cref{thm:full_surplus_multi}, we use the following lemma to verify that the conditions of \cref{lem:full_surplus} are met.

\begin{lemma}\label{lem:pairwise_submodularity}
For any number of products and for any distribution $\dist$ that is negatively affiliated, 
for any $i\neq i'$, we have 
\begin{align*}
\expect[\dist]{ \val_i \given i^*(\val) = i} 
\geq \expect[\dist]{ \val_i \given i^*(\val) = i'}.
\end{align*}
\end{lemma}

For convenience, we first present the proof based on log-supermodularity instead of log-submodularity in this section. Let~$Z$ be a finite distributive lattice, and let $\mu$ be a non-negative function on it. 
Function $\mu$ satisfies log-supermodularity if 
for any $z,z'\in Z$, 
\begin{align*}
\mu(z\wedge z')\cdot\mu(z\vee z')\geq \mu(z)\cdot\mu(z).
\end{align*}
\begin{proposition}[FKG Inequality~\citep{fortuin1971correlation}]
\label{prop:fkg}
For any finite distributive lattice and any non-negative function $\mu$ that satisfies log-supermodularity,
for any non-decreasing functions $g, \hat{g}$ defined on $Z$, 
it holds that 
\begin{align*}
\rbr{\sum_{z\in Z} g(z)\hat{g}(z)\mu(z)}\rbr{\sum_{z\in Z} \mu(z)} \geq \rbr{\sum_{z\in Z} g(z)\mu(z)}
\rbr{\sum_{z\in Z} \hat{g}(z)\mu(z)}.
\end{align*}
\end{proposition}

\begin{corollary}\label{cor:conditional_expectation}
Given any \(\vals \subseteq \mathbb{R}^2\), 
let $\density$ be a probability mass function supported on $\vals$ and satisfying log-submodularity (\cref{def:negative_affiliation}), i.e., \begin{equation}\label{eq:supermodularity}
    f(\min\{\val_1, \val'_1\}, \max\{\val_2, \val'_2\})\cdot f(\max\{\val_1, \val'_1\}, \min\{\val_2, \val'_2\})\geq f(\val)\cdot f(\val').
\end{equation}
Then, the following inequality holds: 
\[(i)~\expect{\val_1 \mid \val_1\geq \val_2}\geq \expect{\val_1} \quad \text{and} \quad (ii)~\expect{\val_1 \mid \val_1\geq \val_2}\geq \expect{\val_1 \mid \val_1< \val_2}.\] 
\end{corollary}
\begin{proof}
For the distributive lattice \(\vals \subseteq \mathbb{R}^2\),
we define the meet and join operations as 
\[
\val \wedge \val' = (\min\{\val_1, \val'_1\}, \max\{\val_2, \val'_2\})
\quad \text{and} \quad
\val \vee \val' = (\max\{\val_1, \val'_1\}, \min\{\val_2, \val'_2\})
\]
for any \(\val, \val' \in \vals\).
According to this definition, distribution $f$ is, in fact, log-supermodular in lattice $V$. 

We choose $g(\val) = \val_1$ and $\hat{g}(\val) = \indicate{\val_1\geq \val_2}$. Using the partial order induced by $\vals$, if $\val\geq \val'$, i.e., $\val\wedge \val'=\val'$, then: (i) $\val_1\geq \val_1'$ and (ii) $\val_2\leq \val_2'$. Thus, both $g$ and $\hat{g}$ are non-decreasing on $\vals$.
Applying \cref{prop:fkg} to our chosen functions $g$ and $\hat{g}$, we have
\begin{align*}
\sum_{\val\in \vals} \val_1\cdot\indicate{\val_1\geq \val_2} \cdot\density(\val) \geq \rbr{\sum_{\val\in \vals} \val_1\cdot\density(\val)}
\rbr{\sum_{\val\in \vals} \indicate{\val_1\geq \val_2} \cdot\density(\val)},
\end{align*}
which implies that 
\[\expect{\val_1 \mid \val_1\geq \val_2}\geq \expect{\val_1} \quad \text{and} \quad \expect{\val_1 \mid \val_1\geq \val_2}\geq \expect{\val_1 \mid \val_1< \val_2}.\]
Therefore, \cref{cor:conditional_expectation} holds.
\end{proof}

\begin{proof}[Proof of \cref{lem:pairwise_submodularity}]
For any pair of products $i\neq i'$, let $F_{-(i,i')}^{\val_{-(i,i')}}$ be the joint distribution over values $(\val_i,\val_{i'})$ 
conditional on the value profile $\val_{-(i,i')}$ for other products and the event that $\max\{\val_i,\val_{i'}\} \geq \max_{j\neq i,i'}\val_j$. 
Since $\dist$ satisfies negative affiliation, by \cref{def:negative_affiliation}, it immediately implies that
distribution $\dist_{-(i,i')}^{\val_{-(i,i')}}$ satisfies negative affiliation as well.\footnote{When the original probability mass function satisfies log-supermodularity, it is not hard to verify that the conditional probability mass function also satisfies log-supermodularity.}
By \cref{cor:conditional_expectation}, an application of the FKG inequality to our setting, it follows easily that
\begin{align*}
\expect[\dist_{-(i,i')}^{\val_{-(i,i')}}]{ \val_i \given \indicate{\val_i\geq\val_{i'}}} 
\geq \expect[\dist_{-(i,i')}^{\val_{-(i,i')}}]{ \val_i} 
\geq \expect[\dist_{-(i,i')}^{\val_{-(i,i')}}]{ \val_i \given \indicate{\val_i<\val_{i'}}}.
\end{align*}
Finally, \cref{lem:pairwise_submodularity} holds by taking the summation over $\val_{-(i,i')}$.
\end{proof}

We are left to prove \cref{thm:full_surplus_exchange}. Similarly, we aim to verify the following lemma.

\begin{lemma}
If \(\dist\) is exchangeable, then for every \(i\neq i'\),
\[
\expect[\dist]{\,\val_i \mid i^*(\val)=i\,}\;\ge\;\expect[\dist]{\,\val_i \mid i^*(\val)=i'\,}.
\]
Hence the hypothesis of \cref{lem:full_surplus} holds for any exchangeable distribution.
\end{lemma}

\begin{proof}
Let \(M:=\max_{j\in[m]} \val_j\). On the event \(\{i^*(\val)=i\}\), the \(i\)-th coordinate attains the maximum, so \(\val_i=M\). Thus
\[
\expect[\dist]{\,\val_i \mid i^*(\val)=i\,}=\expect[\dist]{\,M \mid i^*(\val)=i}.
\]
On the event \(\{i^*(\val)=i'\}\) with \(i'\neq i\), we have \(\val_i\le M\), hence
\[
\expect[\dist]{\,\val_i \mid i^*(\val)=i'\,}\le \expect[\dist]{\,M \mid i^*(\val)=i'\,}.
\]

By exchangeability, the joint law is invariant under relabeling of coordinates. In particular, the distribution of the symmetric statistic \(M\) conditional on \(\{i^*(\val)=i\}\) is the same as that conditional on \(\{i^*(\val)=i'\}\) (a permutation swapping \(i\) and \(i'\) leaves \(M\) unchanged while swapping the two events). Therefore,
\[
\expect[\dist]{\,M \mid i^*(\val)=i\,}=\expect[\dist]{\,M \mid i^*(\val)=i'\,}.
\]
Combining the three displays yields
\[
\expect[\dist]{\,\val_i \mid i^*(\val)=i\,}\;=\;\expect[\dist]{\,M \mid i^*(\val)=i\,}
\;\ge\;\expect[\dist]{\,\val_i \mid i^*(\val)=i'\,},
\]
as claimed.
\end{proof}

{
Moreover, by the exchangeability of $\dist$, and assuming ties are broken
symmetrically, we have
\[
\mathbb{E}_{\dist}\!\left[\,\val_i \mid i^*(\val)=i\,\right]
=
\mathbb{E}_{\dist}\!\left[\,\val_{i'} \mid i^*(\val)=i'\,\right].
\]
It follows that the price is uniform across items: the price is equal to the
posterior mean value of an item, conditional on that item being the
highest-valued one.
}

\subsection{Suboptimality of Pricing Mechanisms}
\label{sub:Suboptimality_of_Item_Pricing}
With at least three products, pricing can be strictly suboptimal even when information is unrestricted. We verify this claim for \cref{exp:lottery_opt}. The lottery serves the lower-value profile more efficiently without increasing the rent of the higher-value profile. By \cref{thm:two_item_optimal}, two products cannot yield such a counterexample.

Consider an instance with three products. The buyer has two different types of values. Specifically, with a probability of half each, the buyer has a value either $\val^{(1)}=(0, 20, 9)$ or $\val^{(2)}=(4, 0, 5)$. 
A feasible mechanism is to fully reveal the values to the buyer, and consider an allocation and payment rule such that 
$\alloc(\val^{(1)}) = (0,1,0), \alloc(\val^{(2)}) = (0.5,0,0.5)$
and $\pay(\val^{(1)}) = 20, \pay(\val^{(2)})=4.5$. 
The expected revenue from this mechanism is $12.25$.

Next, we restrict our attention to pricing mechanisms and show that the optimal revenue for these mechanisms is at most $12$. Note that our joint design problem can be viewed as first fixing a pricing mechanism, and then finding the information structure that maximizes the expected revenue within the given mechanism. 
The latter problem of finding the optimal information structure can be framed as a Bayesian persuasion problem, where the sender's utility function is determined by the pricing mechanism.
By \citet{kamenica2011bayesian}, since the state is binary, it is without loss of optimality to consider only information structures with binary signals. 
This further implies that, in the example of selling three products, it is without loss of optimality to restrict attention to pricing mechanisms that offer only two products for sale.
If product $1$ and $3$ are offered, the optimal welfare under such a mechanism is at most $7$, 
and if product $1$ and $2$ are offered, the optimal welfare under such a mechanism is at most $12$.
In both cases, due to the individual rationality constraint, the revenue of the seller is at most the welfare, which is at most $12$. 
Similarly, the revenue from selling only one product is at most $10$. 

Now we focus on mechanisms that offer only product $2$ and $3$. 
Let $\signals = \{\signal_2,\signal_3\}$ and 
let $p_2, p_3$ be the prices of products $2$ and $3$ respectively. 
Without loss of generality, we assume that the buyer will prefer purchasing product $i$ upon receiving signal $\signal_i$ for $i\in\{2,3\}$. 
Let $\alpha_1,\alpha_2$ be the probabilities that the buyer receives signal $\signal_2$ conditional on the value being $\val^{(1)}$ and $\val^{(2)}$, respectively. 
To ensure that the expected welfare is at least $12$, we must have $\alpha_1 \geq \frac{10}{11}$ and $\alpha_2 \leq \frac{1}{5}$. 
In this case, the posterior value for product $3$ conditional on signal $\signal_2$ is higher than that of signal $\signal_3$. 
Therefore, in the optimal mechanism, it must be the case that 
$$p_3 = \frac{\frac{9}{2}(1-\alpha_1) + \frac{5}{2}(1-\alpha_2)}{\frac{1}{2}(2-\alpha_1-\alpha_2)}$$
and 
$$p_2 = \frac{10\alpha_1}{\frac{1}{2}(\alpha_1+\alpha_2)}
- \frac{\frac{9}{2}\alpha_1+\frac{5}{2}\alpha_2}{\frac{1}{2}(\alpha_1+\alpha_2)} + p_3.$$
The expected revenue of the seller is 
\begin{align*}
R&=\frac{1}{2}(\alpha_1+\alpha_2)p_2 
+ \frac{1}{2}(2-\alpha_1-\alpha_2)p_3 = \frac{11}{2}\alpha_1 - \frac{5}{2}\alpha_2
+ \frac{9(1-\alpha_1) + 5(1-\alpha_2)}{(2-\alpha_1-\alpha_2)}. 
\end{align*}
First, note that 
\begin{align*}
\frac{\partial R}{\partial \alpha_1} 
= \frac{11}{2}-\frac{4(1-\alpha_2)}{(2-\alpha_1-\alpha_2)^2} > 0
\end{align*}
for $\alpha_1\in[\frac{10}{11},1],\;\alpha_2\in[0,\frac{1}{5}]$.
Therefore, the optimal revenue is maximized when $\alpha_1 = 1$, 
which simplifies to 
\begin{align*}
R&= \frac{21}{2}-\frac{5}{2}\alpha_2\leq \frac{21}{2} < 12. 
\end{align*}
Combining all cases, we show that the optimal revenue from pricing mechanisms is at most $12$, which is strictly less than the optimal revenue from lottery pricing. 

\section{Additional Results}
\label{apx:additional}

\subsection{Suboptimality of Horizontal Disclosure}
\label{subapx:horizontal}
If the weakly negative correlation condition in \cref{def:negative_affiliation} is violated, mechanisms with horizontal disclosure can be far from revenue optimal, even if the seller is not restricted to pricing mechanisms. 
The main intuition is that when the values are positively correlated, simply revealing which product has the highest value for the buyer may disclose too much vertical information, allowing the buyer to perfectly infer the realized values of each product. This excessive information disclosure creates significant information rents for the buyer, making such a construction suboptimal.

Specifically, consider an instance with $n$ products and $n$ value profiles. 
Fixing a sufficiently small parameter $\epsilon > 0$. 
For each $i\in [n]$, the probability of the value profile $v^{(i)}$ is $\frac{1}{2^i(1-2^{-n})}$, $v^{(i)}_j = 2^i$ for any $j\neq i$, and $v^{(i)}_i = 2^i + \epsilon$. 

In this instance, if the mechanism discloses which product has the highest value, the buyer can infer the value profile from the highest value product. That is, horizontal disclosure is equivalent to fully revealing information structures in this example. 
Moreover, by fully revealing the values, it is easy to verify that the optimal revenue in this case is at most $\frac{1}{1-2^{-n}}+\epsilon$. 
However, if the seller reveals no information to the buyer, the seller can extract revenue equal to the highest marginal mean value, which is at least $\frac{n}{1-2^{-n}}$.
The multiplicative gap in expected revenue is $n$ when $\epsilon\to 0$. 

\subsection{No Full Surplus Extraction}
\label{subapx:no_full_surplus}
In this section, we provide a simple example to show that full surplus extraction may not be possible in general without the assumption of negative affiliation.
Specifically, consider the example with two products for sale, where the buyer has two possible valuation profiles:

\begin{equation}
    (v_1, v_2) =
    \begin{cases}
      (5,4) & \text{with probability } \frac{1}{2} \\
      (9,10) & \text{with probability } \frac{1}{2}
    \end{cases}
  \end{equation}
If the seller uses horizontal disclosure, the buyer learns exactly which valuation profile he has. Under this information structure, the optimal mechanism is to post a price of $5$ for product 
$1$ and $6$ for product $2$, yielding an expected revenue of $5.5$.

In contrast, by withholding all information and posting a price of $7$ for product $2$ and an infinite price for product $1$, product $2$ is sold with a probability of $1$, generating an expected revenue of $7$, which exceeds $5.5$. Moreover, while this mechanism is optimal for the seller, the expected revenue remains strictly less than the optimal welfare of $7.5$.

\subsection{Alternative Unit-demand Utility Model}
\label{subapx:unit_demand}
In this paper, we assume that for any realized allocation $\alloc\in\{0,1\}^m$, 
the utility of the buyer with a posterior mean value $\belief$ for receiving allocation $\alloc$ while paying price $\pay$ is 
\begin{align*}
\util(\alloc,\pay;\belief) = \max_{i\in[m]} \belief_i\alloc_i - \pay.
\end{align*}
This aligns with applications where the buyer must make a consumption choice before the realization of the values to enjoy the utility.
An alternative model for unit-demand utilities could encompass settings where the buyer always enjoys the maximum realized value after receiving a set of products. 
That is, given the posterior distribution $\mu$ over product values, the utility of the buyer for receiving allocation $\alloc$ while paying price $\pay$ is 
\begin{align*}
\util(\alloc,\pay;\mu) = \expect[\val\sim\mu]{\max_{i\in[m]} \val_i\alloc_i} - \pay.
\end{align*}
In the alternative model, it is straightforward to show that the optimal mechanism reveals no information to the buyer and sells all products at a price equal to the buyer's expected maximum value. This mechanism extracts the full surplus from the buyer. 
However, this approach, which effectively bundles all products in a unit-demand setting, is impractical in most real-world applications. In contrast, when the seller is restricted to pricing mechanisms or uniform pricing mechanisms, all of our approximation results (\cref{prop:uniform_pricing,thm:approx_opt,prop:two_item_approx,thm:full_surplus_multi}) directly generalize in this alternative model. This is because our analysis compares the revenue of pricing mechanisms to the optimal welfare, which remains unchanged across both utility models. Therefore, the conclusion that pricing mechanisms are competitive when the seller can design the information structure continues to hold under the alternative model.

\subsection{Additive Valuations}
\label{subsec:additive}
In our paper, we focus on unit-demand valuations. If the buyer instead has additive valuations, 
that is, for any realized allocation $\alloc\in\{0,1\}^m$, 
the utility of the buyer with posterior value $\belief$ for receiving allocation $\alloc$ while paying price $\pay$ is 
\begin{align*}
\util(\alloc,\pay;\belief) = \sum_{i\in[m]} \belief_i\alloc_i - \pay.
\end{align*}
In this case, pricing mechanisms are always revenue optimal in the joint design problem. 
Specifically, by revealing no information to the buyer and charging the prior mean value for all products, 
the seller can always extract the full surplus using pricing mechanisms, regardless of the value distributions. 

\subsection{Additional Design Power of the Seller}
\label{subapx:additional_design}
In this paper, we limit the seller's design power by restricting the seller from charging prices for information disclosure \citep[c.f.,][]{esHo2007optimal,li2017discriminatory} or setting prices directly contingent on the realized signals. Nonetheless, we show that pricing mechanisms, when paired with a suitable information structure, are already competitive even relative to the benchmark of optimal welfare under these restrictions. In applications where such additional design powers are available to the seller, a suitably designed pricing mechanism---augmented with prices for information disclosure or prices contingent on the true signals---can only generate a weakly higher expected revenue. Moreover, the benchmark of optimal welfare remains unchanged. Therefore, pricing mechanisms are even more competitive in those settings.

\end{document}